%% file: main_arxiv.tex
\documentclass[a4]{scrartcl}

\usepackage[utf8]{inputenc}
\usepackage[T1]{fontenc}
\usepackage{microtype}
\usepackage{lmodern}

\usepackage[fleqn]{amsmath}
\usepackage{amssymb}
\usepackage{amsthm}

\usepackage{cite}

\usepackage[hidelinks]{hyperref}

\usepackage{refcount}
\usepackage{import}
\usepackage{graphicx}
\usepackage{xcolor}
\usepackage[shortlabels]{enumitem}
\usepackage{csquotes,array,multicol,thmtools,xparse,tikz}
\usetikzlibrary{calc}

\usepackage{authblk}

\newtheorem{definition}{Definition}
\newtheorem{theorem}   {Theorem}

\newtheorem{lemma}       [theorem]{Lemma}

\newtheorem{corollary}   [theorem]{Corollary}

\newtheorem*{theorem*}{Theorem}
\newtheorem{manuallemmainner}[theorem]{Lemma}
\newenvironment{manuallemma}[1]{%
\renewcommand\themanuallemmainner{#1}%
\manuallemmainner
}{\endmanuallemmainner\addtocounter{theorem}{-1}}
\newtheorem{manualtheoreminner}[theorem]{Theorem}
\newenvironment{manualtheorem}[1]{%
\renewcommand\themanualtheoreminner{#1}%
\manualtheoreminner
}{\endmanualtheoreminner\addtocounter{theorem}{-1}}

\newcommand{\norm}[1]{\lVert#1\rVert}
\newcommand{\Norm}[1]{\left\lVert#1\right\rVert}

\newcommand{\dotcup}{\mathbin{\dot{\cup}}}

\newcommand{\RouteArg}[4][1pt]{{#3}\xxrightarrow[#1]{#2}{#4}}
\newcommand{\Route}[2]{{#1}\rightarrow{#2}}

\definecolor{niceredbright}{HTML}{bd0310}
\definecolor{nicebluebright}{HTML}{197b9b}
\definecolor{nicered}{HTML}{7f0a13}
\definecolor{niceblue}{HTML}{104354}
\definecolor{nicegreen}{HTML}{217516}
\definecolor{nicepurple}{HTML}{884bab}
\definecolor{nicebg}{HTML}{f6f0e4}
\definecolor{niceredlight}{HTML}{c9888d}
\definecolor{nicebluelight}{HTML}{78a4b8}
\definecolor{nicegreenlight}{HTML}{76de68}
\definecolor{nicepurplelight}{HTML}{bc87db}

\title{Compact Oblivious Routing in Weighted Graphs}

\author{Philipp Czerner, Harald Räcke}
\affil{\{czerner, raecke\}@in.tum.de\\
Department of Informatics, TU München, Germany}

\begin{document}
\def\tildeO{\tilde{\mathcal{O}}}
\def\deg{\operatorname{deg}}
\def\height{\operatorname{height}}
\def\hrnote #1{{\color{red} #1}}
\def\pcnote #1{{\color{nicepurple} #1}}
\def\tildeO{\tilde{\mathcal{O}}}

\maketitle

\begin{abstract}
The space-requirement for routing-tables is an important characteristic of
routing schemes. For the cost-measure of minimizing the total network load
there exist a variety of results that show tradeoffs between stretch and
required size for the routing tables. This paper designs compact routing
schemes for the cost-measure congestion, where the goal is to minimize the
maximum relative load of a link in the network (the relative load of a link is
its traffic divided by its bandwidth). We show that for arbitrary undirected
graphs we can obtain oblivious routing strategies with competitive ratio
$\tildeO(1)$ that have header length $\tildeO(1)$, label size $\tildeO(1)$,
and require routing-tables of size $\tildeO(\deg(v))$ at each vertex $v$ in the
graph.

This improves a result of Räcke and Schmid who proved a similar result in
\emph{unweighted} graphs.
\end{abstract}

\input{intro.tex}

\input{report.tex}

\input{appendix.tex}

\bibliography{references,esa}

\end{document}

%% file: intro.tex
\def\capac{w}
\def\polylog{\operatorname{polylog}}

\section{Introduction}\label{sec:introduction}
\label{sec:intro}

Oblivious routing strategies choose routing paths independent of the traffic in
the network and are therefore usually much easier to implement than adaptive
routing solutions that might require centralized control and/or lead to
frequent reconfigurations of traffic routes. Because of this simplicity a lot
of research in recent years has been performed on the question whether the
quality of route allocations performed by oblivious algorithms is comparable to
that of adaptive solutions (see e.g.\
\cite{azar2004optimal,borodin1985routing,ip-ton,kodialam2009oblivious,semi,racke2002minimizing,racke2008optimal,towles2002worst}). For some cost-metrics this is indeed the case. For
example for minimizing the total traffic in the network (a.k.a. total load),
shortest path routing is a simple optimal oblivious strategy. When one aims to
minimize the congestion, i.e., the maximum (relative) load of a network link,
one can still obtain strategies with a competitive ratio of $\mathcal{O}(\log n)$, i.e.,
the congestion generated by these strategies is at most an $\mathcal{O}(\log n)$-factor
than the best possible congestion~\cite{racke2008optimal}.

However, another important aspect for implementing oblivious routing strategies
on large networks is the size of the required routing tables. This aspect has
been investigated thoroughly for the cost-measure total load~(see e.g.\
\cite{cowen2001compact,fraigniaud1995memory,frederickson1988designing,krioukov2004compact,retvari2013compact,thorup2001compact,van1995compact}), and various
trade-offs between competitive ratio (also called stretch for the total load
scenario) and the table-size have been discovered.

If for example every vertex stores the next hop on a shortest path to a target
one can obtain a stretch of $1$ at the cost of having routing tables of size
$\mathcal{O}(n\log n)$ per node. If one allows non-optimal solutions Thorup and Zwick~\cite{thorup2001compact}
have shown how to obtain a stretch of $4k-5$ for any $k>2$ with routing tables
of size $\tildeO(n^{1/k})$. This routing scheme works for the so-called
labeled scenario in which the designer of the routing-scheme is allowed to
relabel the vertices of the network in order to make routing decisions easier.
Of course, there is still a restriction on the label-size as otherwise
the power of being able to assign labels to vertices could be abused.

In the (more difficult) so-called \emph{name-independent} model the designer is
not allowed to relabel the vertices. Abraham et al.~\cite{ittai} have shown that
for general undirected graphs one can asymptotically match the bounds for the
labeled variant. They obtain a stretch of $\mathcal{O}(k)$ and routing tables of size
$\tildeO(n^{1/k})$. If a famous conjecture due to Erd\H{o}s~\cite{Erd63} about
the existence of low-girth graphs holds then there is also a lower bound that
says that obtaining a stretch better than $2k+1$ requires routing tables of
size $\Omega(n^{1/k})$. This means that for general undirected graphs the
existing tradeoffs between stretch and space are fairly tight. 

There exist many more results that analyze problem variants as e.g.\ obtained
by restricting the graph representing the network (see
e.g.~\cite{krioukov2004compact,cowen2001compact,fraigniaud2001routing,gavoille2001routing,gavoille1996memory,retvari2013compact});
so the problem of designing compact routing schemes is very well studied for
the cost-measure \emph{total load}.

However, for the cost-measure congestion this is not the case. Räcke and
Schmid~\cite{raecke2018compact} gave the first oblivious routing scheme that combines
a guarantee w.r.t.\ the congestion with small routing-tables. They consider
the labeled model and design an oblivious routing scheme that for a general
undirected, unweighted graph $G$ requires routing tables of size
$\tildeO(\deg(v))$ at each vertex $v$ and obtains a competitive ratio of
$\tildeO(1)$ w.r.t.\ congestion.

There are important differences when comparing this result to its counter-parts
for the total-load scenario. Firstly, the space used at a vertex $v$ may depend on the
degree of $v$. This is a reasonable assumption from a practical perspective as
a node corresponds to a router in the network and it is reasonable to assume
that the memory at a router (node) grows with the number of ports (number of
incident edges). However, this assumption seems also crucial for getting any
reasonable guarantees. In order to minimize congestion it is important to
distribute the traffic among all network resources. It seems very difficult to
do this if the routing table at a vertex is a lot smaller than the number of
outgoing edges.

Another difference is that there is no tradeoff parameter $k$ that gives a
smooth transition from optimal routing with large tables to more compact
routing. The reason is that for \emph{congestion} the competitive ratio may be
$\Omega(\log n)$ even for unlimited routing tables~\cite{bartal1997line}.

One important shortcoming of the result by Schmid and Räcke~\cite{raecke2018compact} is
that it only applies to unweighted graphs (there is a straightforward
generalization that obtains routing tables of size $\mathcal{O}(W\polylog n)$ where $W$
is the largest weight of an edge, but this is undesirable). This restriction is
%
due to the fact that the result by Schmid and Räcke uses \emph{paths} to route within
well-connected clusters, which they obtain by randomized rounding.

There are two major obstacles in generalizing this result to the weighted case:
\begin{enumerate}[(1)]
\item In unweighted graphs, low congestion also ensures a low number of paths using an edge. However, an edge of weight $W$ might be used by $W$ small paths, which cannot be stored in a compact manner.
\item Even within a well-connected cluster, it is not sufficient to route a commodity using a small number of paths, if the nodes are connected by many low-weight edges (illustrated in Figure \ref{fig:bigoversmall}). Hence a source node may have to route (and store) $W$ small paths.
\end{enumerate}

In this paper we give a construction of an oblivious routing scheme that avoids both problems, by storing aggregate routing information for many paths at once, as well as distributing storage across nodes for commodities that need to spread out over multiple paths. In this manner we obtain a
polylogarithmic competitive ratio with polylogarithmic space requirement per
edge in the network. Our main result is the following.

\begin{manualtheorem}{\ref{thm:finalresult}}
There is a compact oblivious routing scheme with competitive ratio at most
$\mathcal{O}(\log^6n\log^3W)$, that uses routing tables of size
$\mathcal{O}(\log^5n\log W\log^3(nW)\cdot\deg(v))$ at a node $v$, packet
headers of length $\mathcal{O}(\log^3(nW))$, and node labels of length
$\mathcal{O}(\log^2n)$.
\end{manualtheorem}
In particular this result shows that if we can route some demand in a network
with a multicommodity flow $f$ of congestion $C$, then it is possible to route the
demands \emph{space-efficiently}, i.e., one can set up small routing tables so that
packets follow a (maybe) different flow $f'$ that routes the same demands with
a slightly worse congestion.
This question of space-efficiently routing demands in a network is orthogonal
to oblivious routing and it is not clear by how much the performance (i.e., the
congestion) degrades because of the space-requirement. The above theorem gives
a polylogarithmic upper bound but to the best of our knowledge this problem has
not been studied before.

\begin{figure}[t]
\begin{center}
\def\svgwidth{12.7cm}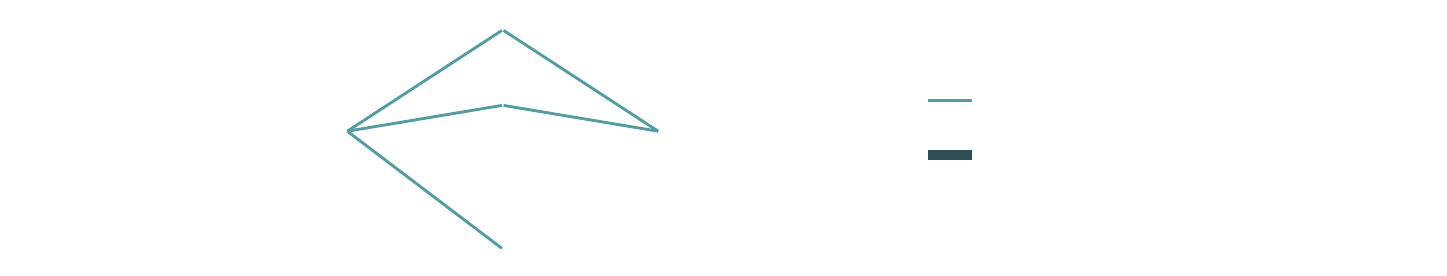
\end{center}
\caption{Routing a single demand over multiple edges. Sending data from $a$ to $b$ requires roughly $W$ paths, but $a$ has degree 1 and can store only $\tildeO(1)$ bits.}
\label{fig:bigoversmall}
\end{figure}

\subsection{Further Work}
Oblivious routing with the goal of either minimizing the total load (or
stretch), minimizing the congestion or a combination of both is a well studied
problem. The research started with deterministic algorithms and it was shown
by Borodin and Hopcroft~\cite{borodin1985routing} that on \emph{any} bounded degree
graph $G$ for \emph{any} deterministic routing scheme there exists a permutation
routing instance that incurs congestion $\Omega(\sqrt{n}/\Delta^{3/2})$. 
This result was improved by Kaklamanis et al.~\cite{kaklamanis1991tight}
to a lower bound of $\Omega(\sqrt{n}/\Delta)$. As there exist bounded degree
graphs that can route any permutation with small congestion this gives a large lower
bound on the competitive ratio of deterministic oblivious routing schemes.

For randomized algorithms Valiant and Brebner~\cite{VB81} showed how to obtain
a polylogarithmic competitive ratio for the hypercube by routing to random
intermediate destinations (known as Valiant's trick).

R{\"a}cke~\cite{racke2002minimizing} presented the first oblivious routing scheme with a
polylogarithmic competitive ratio of $\mathcal{O}(\log^3{n})$ in general undirected
networks. This routing scheme is based on a hierarchical decomposition of
a graph and forms the basis for the compact routing schemes that we construct
in this paper.  The construction in~\cite{racke2002minimizing} was not polynomial time. This
drawback was independently addressed by Bienkowski et
al.~\cite{bienkowski2003practical} and Harrelson
et al.~\cite{harrelson2003polynomial}. Both papers give a polynomial-time
algorithm for constructing the hierarchical decomposition (and, hence, the
routing scheme)--- the first with a competitive ratio of $\mathcal{O}(\log^4n)$, and
the second with a competitive ratio of $\mathcal{O}(\log^2n\log\log n)$.

In 2014 R{\"a}cke et al.~\cite{racke2014computing} presented another
construction of the hierarchy that runs in time $\mathcal{O}(m\polylog n)$ and guarantees
a competitive ratio of $\mathcal{O}(\log^5 n)$ (however, going from the hierarchy to the
actual routing scheme may require superlinear time).

The above oblivious routing schemes that are based on hierarchical tree
decompositions do not give the best possible competitive ratio.
In~\cite{racke2008optimal} Räcke presents an oblivious routing scheme that is
based on embedding a convex combination of trees into the graph $G$. This
scheme obtains a competitive ratio of $\mathcal{O}(\log n)$, which is optimal due to a
lower bound of Bartal and Leonardi for online routing in
grids~\cite{bartal1997line}.

However, the number of trees that are used in the above
result~\cite{racke2008optimal} is fairly large ($\Theta(m)$). Therefore, it
seems difficult to design a compact routing scheme based on the tree embedding
approach, and, therefore we use the earlier results that are based on
hierarchical decompositions (a single tree!) but only guarantee slightly weaker
competitive ratios.

%% file: figures/bigsmallsimple.pdf_tex
\begingroup%
  \makeatletter%
  \providecommand\color[2][]{%
    \errmessage{(Inkscape) Color is used for the text in Inkscape, but the package 'color.sty' is not loaded}%
    \renewcommand\color[2][]{}%
  }%
  \providecommand\transparent[1]{%
    \errmessage{(Inkscape) Transparency is used (non-zero) for the text in Inkscape, but the package 'transparent.sty' is not loaded}%
    \renewcommand\transparent[1]{}%
  }%
  \providecommand\rotatebox[2]{#2}%
  \newcommand*\fsize{\dimexpr\f@size pt\relax}%
  \newcommand*\lineheight[1]{\fontsize{\fsize}{#1\fsize}\selectfont}%
  \ifx\svgwidth\undefined%
    \setlength{\unitlength}{416.69291339bp}%
    \ifx\svgscale\undefined%
      \relax%
    \else%
      \setlength{\unitlength}{\unitlength * \real{\svgscale}}%
    \fi%
  \else%
    \setlength{\unitlength}{\svgwidth}%
  \fi%
  \global\let\svgwidth\undefined%
  \global\let\svgscale\undefined%
  \makeatother%
  \begin{picture}(1,0.19047619)%
    \lineheight{1}%
    \setlength\tabcolsep{0pt}%
    \put(0,0){\includegraphics[width=\unitlength,page=1]{bigsmallsimple.pdf}}%
    \put(0.69222424,0.11274138){\color[rgb]{0,0,0}\makebox(0,0)[lt]{\lineheight{1.25}\smash{\begin{tabular}[t]{l}edge (weight 1)\end{tabular}}}}%
    \put(0.69222424,0.07509411){\color[rgb]{0,0,0}\makebox(0,0)[lt]{\lineheight{1.25}\smash{\begin{tabular}[t]{l}edge (weight $W$)\end{tabular}}}}%
    \put(0,0){\includegraphics[width=\unitlength,page=2]{bigsmallsimple.pdf}}%
    \put(0.53549454,0.1286087){\color[rgb]{0,0,0}\makebox(0,0)[lt]{\lineheight{1.25}\smash{\begin{tabular}[t]{l}$b$\end{tabular}}}}%
    \put(0.14611235,0.1286087){\color[rgb]{0,0,0}\makebox(0,0)[lt]{\lineheight{1.25}\smash{\begin{tabular}[t]{l}$a$\end{tabular}}}}%
  \end{picture}%
\endgroup%

%% file: report.tex

\def\bold #1{{\bfseries\mathversion{bold}#1}}

\def\deg{\operatorname{deg}}
\def\height{\operatorname{height}}
\def\bal{\operatorname{bal}}
\def\out{\operatorname{out}}
\def\strich{$'$}
\def\ell{l}

\NewDocumentCommand{\xxrightarrow}{ O{}m }{%
\mathrel{%
\vcenter{\hbox{%
\begin{tikzpicture}%
  \node[anchor=south,align=center,inner sep=0pt] (a) {};
  \node at (a) [minimum width=0.6cm,anchor=south,overlay,inner ysep=1pt, inner xsep=4pt,yshift=-#1] (b) {$\scriptstyle #2$};
  \coordinate (left)  at ($(b.south west)+(-0.0cm,0cm)$);
  \coordinate (right) at ($(b.south east)+ (0.0cm,0cm)$);
  \draw[->,overlay] (left) -- (right);
  \path (left |- a) -- (right |- a);
\end{tikzpicture}%
}}%
}%
}
\def\hrnote #1{{\color{red} #1}}
\def\pcnote #1{{\color{nicepurple} #1}}
\def\tildeO{\tilde{\mathcal{O}}}




\def\cong{\operatorname{cong}}
\newcommand{\ca}{w} 
\newcommand{\we}{c} 
\newcommand{\Co}{C} 
\newcommand{\Sc}{x} 
\newcommand{\bwe}{\bar{\we}} 
\newcommand{\Nclass}{N_\mathrm{class}} 

\subsection{Preliminaries}

Throughout the paper we use $G=(V,E,\ca)$ to denote an undirected weighted graph
with $n$ node and $m$ edges. We will
refer to the weight of an edge also as the \emph{capacity} of the edge. Wlog.\
we assume that the minimum edge weight is $1$, that edge-weights are a power of
$2$, and that the largest edge weight is $W$.
We call an edge of capacity/weight $2^i$ a \emph{class} $i$ edge and use 
$\Nclass:=1+\log_2W$ to denote the total number of classes.
Further, we use $\Gamma(v)$ to denote the neighborhood of a vertex $v$, i.e.,
$\Gamma(v)=\{u\in V\mid \{u,v\}\in E\}$.

The degree of a node $v$ in the graph $G$ will be referred to as $\deg_G(v)$,
that is $\deg_G(v):=|\Gamma(v)|$. We apply that to directed graphs as
well, where it refers to the number of outgoing edges.

%

\def\Eor{E_{\mathrm{or}}}

While the edges $E$ are undirected, it will be convenient to refer to a certain
orientation of an edge, so we define $\Eor:=\{(u,v)\in V^2:\{u,v\}\in E\}$.
%
%
A mapping $f:\Eor\rightarrow\mathbb{R}$ with $f((u,v))=-f((v,u))$ for
$(u,v)\in \Eor$ is called a (single-commodity) flow. If $f(u,v)>0$
for some edge $(u,v)\in E$, this indicates flow from $u$ to $v$. The reverse
flow of $f$ is simply $-f$. For the sake of readability we
omit double parentheses and 
write, e.g., $f(u,v)$ instead of $f((u,v))$. 

A flow $f$ may have multiple sources and sinks. The balance of a node $v\in V$
is denoted by $\operatorname{bal}_f(v):=\sum_{u\in\Gamma(v)}f(u,v)$, so a positive
balance indicates that the node is receiving more flow than sending out. A flow
is \emph{acyclic}, if there is no path $(p_0,...,p_k)$ in $G$ with $p_0=p_k$
and $f(p_i,p_{i+1})>0$ for all $i$.
Its congestion is the maximum ratio between the flow over an edge and its
weight, denoted by $\cong(f):=\max_{\{u,v\}\in E} |f(u,v)|/\ca(u,v)$.
Given a multi-set of flows $F:=\{f_1, f_2, ..., f_k\}$, its \emph{total congestion} is
$\cong(F):=\max_{\{u,v\}\in E}\sum_{k} |f_k(u,v)|/w(u,v)$.

If a flow $f$ has all flow originating at a single node $s$, i.e.,
$\bal_f(s)\le 0$ and $\bal_f(u)\ge 0$ for $u\ne s$, we say that $f$ is an
$s$-flow. If additionally $\bal_f(s)=-1$, we call $f$ a \emph{unit} $s$-flow.
The set of all unit $s$-flows is denoted with
$\operatorname{flow}(s)$. If a flow $f$ only sends from $s$ to $t$, i.e.,
$f$ is an $s$-flow and $-f$ is a $t$-flow we call $f$ an $s$-$t$ flow.

We use $\tildeO$ to disregard logarithmic factors, so $g=\tildeO(h)$ iff
$g=\mathcal{O}(h\log^c(nW))$ for some constant~$c$.

\paragraph*{Oblivious Routing Scheme}
Now we define the concept of an oblivious routing scheme. The idea is to fix a
single flow between each pair of nodes $(u,v)$, and then multiply that flow
with the actual demand from $u$ to $v$ to get the route. This flow can be
interpreted probabilistically or fractionally, so if we have $f(e)=\frac{1}{2}$
for some edge $e$ it means that the probability of the packet being routed
across edge $e$ is $\frac{1}{2}$; or that half a packet travels along that
edge. We will use both interpretations interchangeably.

\begin{definition}
An \emph{oblivious routing scheme} $S=(f_{u,v})_{u,v\in V}$ consists of a unit
$u$-$v$-flow for each pair of nodes $u,v\in V$.
Given demands $d:V\times V\rightarrow\mathbb{R}_{\ge 0}$ the congestion of $S$
w.r.t.\ $d$, denoted $\cong(S,d)$, 
is the total congestion of the set of flows $\{d(u,v)f_{u,v}:u,v\in V\}$.
The \emph{competitive ratio} of $S$ is
$\max_{d}{\cong(S,d)}/{\cong_{\mathrm{opt}}(d)}$, where
$\cong_{\mathrm{opt}}(d)$ denotes the optimal congestion that can be obtained
for demands $d$ by any routing scheme.
\end{definition}

Defining a compact oblivious routing scheme formally is a bit more
involved, as we have to clarify where information is stored and how it is used.
Before we do so, we introduce the notion of a \emph{routing algorithm}, which defines formally how packets are sent through the network. The intuition is that each packet carries a \emph{packet header}, storing per-packet information. Each node stores a \emph{routing table}, containing arbitrary information for the routing algorithm to use.

The routing algorithms forwards a packet in a local manner, meaning that it reads both the packet header and the routing table before choosing an outgoing edge on which to sent the packet. At the same time, it may modify the packet header. This procedure repeats, until routing algorithm indicates that the packet has reached its destination, by outputting no outgoing edge.

It remains to describe how the packet header is initialized. For oblivious routing schemes, we simply use the name of the target node as initial packet header. However, we will later define more general building blocks, namely transformation schemes. Hence a routing algorithm works with a set of abstract \emph{input labels} as possible initial packet headers (and thus as input to the routing algorithm). For the purposes of a routing algorithm these are simply some set, but later definitions will describe their structure more concretely.

\begin{definition}
A \emph{routing algorithm} $A=(\mathcal{A},\mathcal{L},\mathcal{T})$ is a tuple, $\mathcal{L}\subset\{0,1\}^*$ denoting a finite set of \emph{input labels} and $\mathcal{T}:V\rightarrow\{0,1\}^*$ a \emph{routing table} for each node. Additionally, $\mathcal{A}:\mathcal{T}(V)\times\{0,1\}^*\rightarrow (E\cup\{\emptyset\})\times\{0,1\}^*$ is a (possibly randomized) algorithm, taking both a routing table and a packet header as input, which calculates both the outgoing edge (if any) and the new packet header.
\end{definition}

We remark that the outgoing edge given by $\mathcal{A}$ has to be encoded in some manner, and it must be adjacent to the node the routing table belongs to. The routing table $\mathcal{T}(v)$ for a node $v$ can contain information about $v$, such as a list of adjacent nodes, so any straightforward encoding, e.g., the index in this list, will work.

Given a routing algorithm $A=(\mathcal{A},\mathcal{L},\mathcal{T})$, a start vertex $v\in V$, and an input label $\ell\in\mathcal{L}$, the above mechanism defines a process for probabilistically distributing packets from $v$ to targets in the network. We use $A(v,\ell)$ to denote a flow that describes the associated routing paths. This is defined as follows: We inject a packet at $v$, with $\ell$ as packet header and execute $\mathcal{A}$ until no outgoing edge is returned. Then $A(v,\ell)(e)$ is the probability that the packet is routed over $e$ (note that $\mathcal{A}$ may be randomized).

A routing algorithm $A=(\mathcal{A},\mathcal{L},\mathcal{T})$ is \emph{compact}, if packet headers and input labels have size $\tildeO(1)$, and the routing table of a node $v\in V$ has size $\tildeO(\deg(v))$.
\medskip

Recall that an oblivious routing scheme corresponds to a routing algorithm where the input labels are names of nodes in the graph. Consequently, we say that such a scheme is compact if its routing algorithm is compact.

Formally, we assign a name to each vertex in the graph, which we call \emph{node label}, i.e., we have some function $\operatorname{node}:V\rightarrow\{0,1\}^*$. For an oblivious routing scheme $S=(f_{u,v})_{u,v\in V}$ we use the set of all node labels $\operatorname{node}(V)$ as input labels, so the initial packet header is the node label of the target node. 

We say that $S$ is \emph{compact} if there exists a compact routing algorithm $A=(\mathcal{A},\operatorname{node}(V),\mathcal{T})$ with $f_{u,v}=A(u,\operatorname{node}(v))$. This definition matches the one used by Räcke and Schmid~\cite{raecke2018compact}, although it is more explicit.

The main result of this paper is the existence of a compact oblivious routing scheme, with competitive ratio $\tildeO(1)$.

\def\one{\text{\mathversion{bold}$1$\mathversion{normal}}}

\paragraph*{Transformation Schemes}
Our routing scheme will be composed out of several building blocks, which we
call \emph{transformation schemes}. Loosely speaking, they correspond to single-commodity flows which we are able to route.

We consider \emph{distributions} or \emph{weight functions} of the form
$\mu:V\rightarrow\mathbb{R}_{\ge0}$ that assign a non-negative weight
to vertices in $V$. If we only specify a weight function on a subset
$S\subseteq V$ we assume that it is $0$ on $V\setminus S$. We use
$\mu(S):=\sum_{v\in S}\mu(v)$ to denote the weight of a subset $S$, and
$\one_v:V\rightarrow\mathbb{N}$ to denote the special 
weight function that has weight $1$ on $v$ and $0$ elsewhere.
For a distribution $\mu$ we use $\bar{\mu}:=\frac{1}{\mu(V)}\mu$ to denote the
corresponding \emph{normalized distribution}.

\def\TS{\mathit{TS}}
\begin{definition}
A \emph{(compact%
\footnote{As all of our transformation schemes are compact (the later variants of deterministic and concurrent transformation schemes will be as well), we may drop the ‘compact’ when appropriate.}%
) transformation scheme (TS)} is a compact routing algorithm with a single input label.
\end{definition}

\def\muin{\mu_{\operatorname{in}}}
\def\muout{\mu_{\operatorname{out}}}
\def\bmuin{\bar{\mu}_{\operatorname{in}}}
\def\bmuout{\bar{\mu}_{\operatorname{out}}}

The above definition is not very useful by itself. The underlying idea is that we view a
transformation scheme $\TS$ as on operation to transform one distribution of packets into another,
by executing the routing algorithm. More precisely, given $P$ packets each
packet follows the flow $\TS(v)$ at its source location $v\in V$.
This will send it to some target node (probabilistically, according to $\TS(v)$).

We say that a transformation scheme routes from some \emph{input distribution}
$\muin$ to an \emph{output distribution} $\muout$, if the above process
transforms a set of $P$ packets that are distributed according to $\bmuin$
(i.e., a node $v$ has $\muin(v)/\muin(V)\cdot P$ packets in expectation) into
a set of packets that are distributed according to $\bmuout$, i.e., afterwards
a node has $\muout/\muout(V)\cdot P$ packets in expectation.

In addition we associate a \emph{demand} $d(\TS)$ and \emph{congestion}
$\cong(\TS)$ with a transformation scheme $\TS$. We say a tranformation scheme
routes demand $d(\TS)$ from $\muin$ to $\muout$ with congestion $\cong(\TS)$ if the
expected load on an edge $e$ for the above process is at most
$\cong(\TS)\cdot \ca(e)$ when $P=d(\TS)$ (we allow $P$ to be
non-integral).

Note that, of course, the input for a transformation scheme could be any packet
distribution. However, the congestion of the scheme
is stated w.r.t.\ some fixed input distribution $\muin$ (its \emph{natural
input distribution}) and some total demand $d(\TS)$.

From the congestion-value for $\muin$ and its demand $d$, one can
then deduce the congestion-value for other inputs. If we, e.g., use the
transformation scheme on a demand $d'$ that is distributed according to $\nu$
we experience congestion at most $\max_{v\in V}d'\bar{\nu}(v)/d\bmuin(v)$.

To make our notation more concise, we write a statement like
\enquote{\emph{$\TS$ routes $\muin$ to $\muout$ with demand $d$ and congestion
    at most $C$}} as \enquote{\emph{$\TS$ routes
    $\RouteArg[1pt]{d}{\muin}{\muout}$} with congestion (at most) $C$}. We omit
the demand $d$ if it equals $1$.

It may happen that for some transformation scheme $\TS$
we cannot exactly specify the output distribution 
that corresponds to its natural input distribution $\muin$.
We say $\TS$ routes $\RouteArg[1pt]{}{\muin}{\muout}$ with
\emph{approximation} $\sigma$ if the real output distribution $\muout'$
fulfills
$\bmuout(v)/\sigma\le\bmuout'(v)\le\bmuout\cdot\sigma$.

Finally, in some proofs we will view packets as discrete entities and specify that the transformation scheme does not split them up. 

However, this collides with the fractional nature of the transformation scheme, which is caused by randomization. Therefore we introduce the following definition of a deterministic transformation scheme, that extracts this randomness and makes it explicit.

\begin{definition}
A \emph{(compact) deterministic transformation scheme} $\TS=(\mathcal{A},\mathcal{P}(V),\mathcal{T})$ is a compact routing algorithm where $\mathcal{A}$ is deterministic and $\mathcal{P}(v)=\{1,...,N_v\}$ is a set of \emph{path-ids} valid for a node $v\in V$.
\end{definition}

The idea of the above definition is that we can specify a “random seed” as input label, which will determine precisely how a packet is routed. The ordinary transformation scheme will correspond to choosing an input label u.a.r.\ from sets $\mathcal{P}(v)$. In this sense the above definition makes the random choices of a transformation scheme explicit.

Note that, technically, the definition of routing algorithms allows any path id in $\mathcal{P}(V)$ to be specified at a node $v$, not only the ones in $\mathcal{P}(v)$. We ensure that this does not occur.

As $\mathcal{A}$ is deterministic, the path id indeed determines the exact route a packet will take when starting at a certain node. More precisely, each flow $\TS(v,\ell)$ for $v\in V,\ell\in\{1,...,N_v\}$ is simply a path starting at $v$. Still, there is no guarantee that different path-ids send the packet to the same node.

We associate a transformation scheme with each deterministic TS, by choosing the path-id uniformly at random. In this fashion we extend the notions, such as congestion, input/output distributions, etc., that were defined above for ordinary transformation schemes to also cover deterministic transformation schemes.

\paragraph*{Concurrent Transformation Schemes}

While a transformation scheme can mix packets arbitrarily, often we want to distribute several commodities at the same time, with separate input and output distributions for each commodity. This allows us to analyze the congestion more precisely and aggregate the routing information for different commodities. Hence we define the notion of a \emph{concurrent transformation scheme}, which executes multiple transformation schemes in parallel.

The idea is that we take a transformation scheme and additionally specify a commodity as input. 

\begin{definition}
Let $I$ denote a set of commodities.
A \emph{(compact) concurrent [deterministic] transformation scheme (CTS)} is a compact routing algorithm
$\TS=(\mathcal{A},I\times\mathcal{L},\mathcal{T})$, s.t.\ $\TS_i:=(\mathcal{A},\{i\}\times\mathcal{L},\mathcal{T})$ is a [deterministic] transformation scheme for each commodity $i\in I$.
\end{definition}

Note that transformation schemes have a single input label, in which case the $\mathcal{L}$ in the above definition is superfluous and the input to the concurrent transformation scheme is just the commodity. If it is deterministic, we need the path id as input, and $\mathcal{L}$ would be the set of possible path ids. Similar to before, any combination of commodity and path id may be specified at a node, according to the definition of a routing algorithm, but for our purpose only some of these make sense.

The congestion of such a concurrent transformation scheme is defined as
follows. Let $\mu_i$ and $d_i$ denote the input distribution and demand of
$\TS_i$, respectively. Let $X_i(e)$ denote the expected load on an edge $e$ if we
execute $\TS_i$ on $d_i$ packets distributed according to $\mu_i$. The
congestion of the CTS $D$ is defined as
$\cong(D):=\max_{e}\frac{1}{w(e)}\sum_iX_i(e)$.

As an input to the routing algorithm, a commodity $i\in I$ has to be encoded in some fashion. Often, the commodity is determined by the source node (i.e., for each node $v$ at most one input distributions $\mu_i$ is nonzero) and does not need to be specified. Otherwise, we will explicitly describe the necessary encoding as a property of the CTS.

Analogous to transformation schemes, we write \enquote{\emph{$\TS$ routes $\RouteArg[1pt]{d_i}{\mu_i}{\nu_i}$} with congestion (at most) $C$} for each commodity $i\in I$ for a CTS, and extend this notation to deterministic CTS by considering the associated transformation schemes.

\section{Overview}
\label{sec:overview}
In this section we give a high-level overview of the most important steps in
our construction. The first part gives a rough outline of the general approach
of routing along a decomposition tree that forms the basis for some oblivious
routing schemes (e.g.~\cite{racke2002minimizing, bienkowski2003practical,
  harrelson2003polynomial}), and has also been used by Räcke and
Schmid~\cite{raecke2018compact} to obtain compact routing schemes.

\subsection{The Decomposition Tree}
\label{sec:decomp}
The result by Räcke and Schmid~\cite{raecke2018compact} as well as our extension of it use a
decomposition tree, in particular the one described in
\cite{racke2002minimizing}. We refer the reader to these for a more
detailed description and just briefly mention the key ideas here.
We start with a single cluster containing all nodes, and then further refine
that until all clusters consist of just a single node. Hence we get a tree $T$
where nodes are subsets of $V$, which we call \emph{clusters}. The tree $T$ has
root $V$, i.e., the cluster containing all nodes, and leaves $\{v\}$ for each
$v\in V$. For a cluster $S$ with children $S_1,...,S_r$ we have
$S=S_1\dotcup...\dotcup S_r$, so the children are a partition of the parent.
We use $\height(T)$ to denote the maximum distance from any leaf to
the root, and $\deg(T)$ to denote the largest number of children of any cluster.

Now we introduce a number of distributions, which will be
important for routing within the decomposition tree. For any cluster $S$ we
define the \emph{border-weight} $\out_S:S\rightarrow\mathbb{N}$
by $\out_S(v):=\sum_{u\notin S}\ca(v,u)$ for $v\in S$, counting the total
weight of edges leaving the cluster adjacent to a node. Additionally, for any
cluster $S$ with child clusters $S_1,...,S_r$ we define the
\emph{cluster-weight} $w_S:S\rightarrow\mathbb{N}$ as
$w_S:=\sum_i \out_{S_i}$, which also counts edges between
children of $S$. These distributions are shown schematically in Figure
\ref{fig:distributions}.

\begin{figure}
\newcommand{\figdistributionswidth}{4.2cm}
\begin{center}
\begin{minipage}{\figdistributionswidth}
\begin{center}
\def\svgwidth{\figdistributionswidth}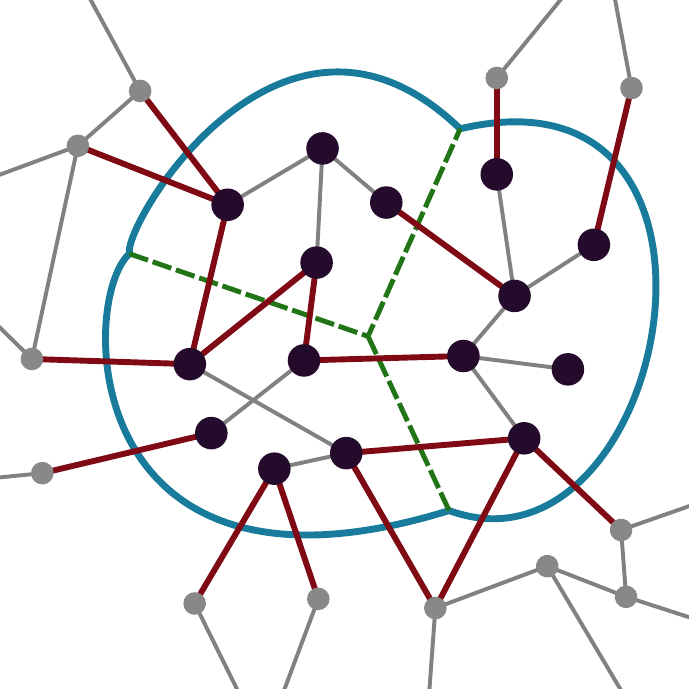
\end{center}
\centering $w_S$
\end{minipage}\hspace{3mm}
\begin{minipage}{\figdistributionswidth}
\begin{center}
\def\svgwidth{\figdistributionswidth}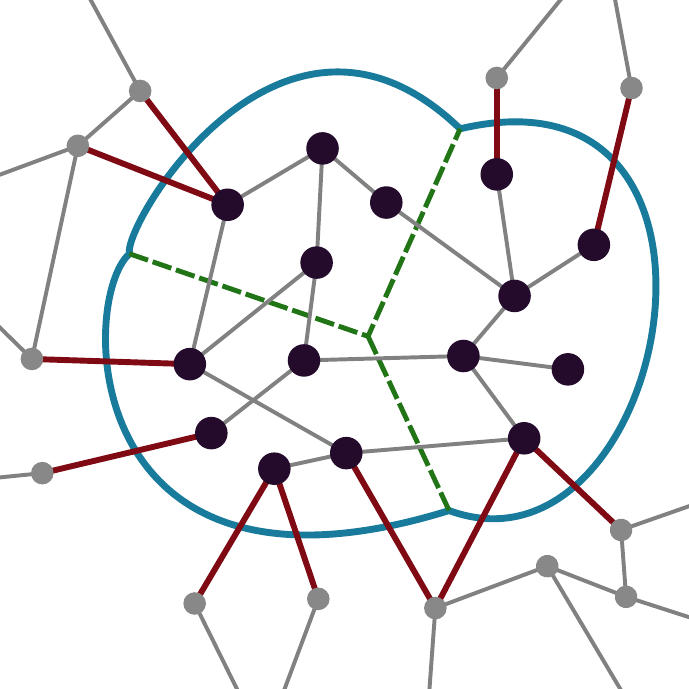
\end{center}
\centering $\out_S$
\end{minipage}\hspace{3mm}
\begin{minipage}{\figdistributionswidth}
\begin{center}
\def\svgwidth{\figdistributionswidth}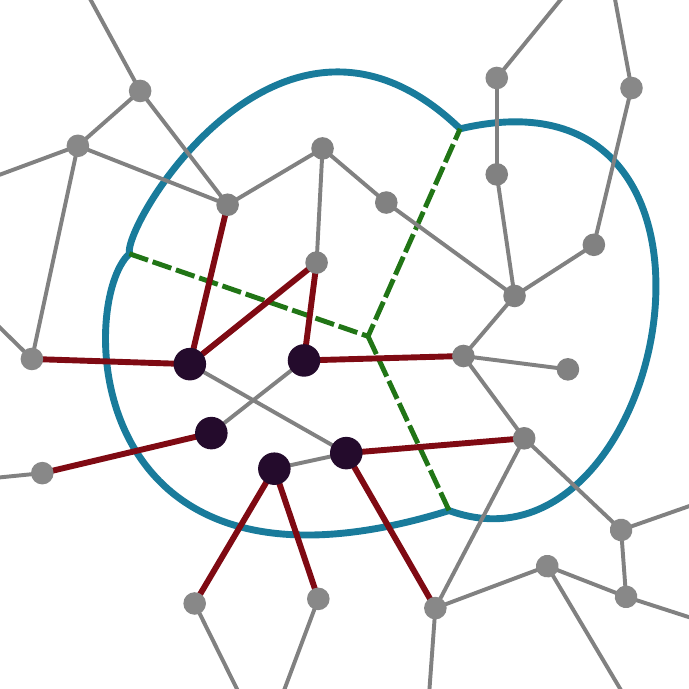
\end{center}
\centering $\out_{S_i}$
\end{minipage}
\end{center}
\caption{Distributions inside a cluster $S$ with child $S_i$. The child
  clusters are separated by dashed lines. The respective distribution is
  nonzero only on the highlighted nodes and counts the total weight of
  highlighted edges adjacent to a node.}
\label{fig:distributions}
\end{figure}

The decomposition from \cite{racke2002minimizing} has two essential
properties:
\begin{itemize}
\item For each cluster $S$ we can solve the multi-commodity flow problem with
demands $d(u,v):=w_S(u)w_S(v)/w_S(S)$ for $u,v\in S$ with congestion
$\Co\in\mathcal{O}(\log^2 n)$ \emph{within} $S$, and
\item the tree has logarithmic height, i.e., $\height(T)\in\mathcal{O}(\log n)$.
\end{itemize}

The essential idea for oblivious routing is that in order to route between two
nodes $s$ and $t$ in the graph we determine the path
$\{s\}=S_1,S_2,\dots,S_k=\{t\}$ in the tree and then we route a packet
successively along this path by routing it from distribution $\bar{w}_{S_i}$ to
distribution $\bar{w}_{S_{i+1}}$ for $i=1,\dots,k-1$. Note that distribution
$\bar{w}_{S_1}=\bar{w}_{S_{\{s\}}}=\one_s$ and distribution 
$\bar{w}_{S_k}=\bar{w}_{S_{\{t\}}}=\one_t$, i.e., we indeed route from $s$ to
$t$.

Now suppose that the optimal congestion for the given demand $d$ is
$C_{\mathrm{opt}}(d)$. How much demand does the above process induce for
routing from $w_{S_i}$ to $w_S$ for a child-cluster $S_i$ of some cluster $S$
(for all packets)? 
Each packet that uses the tree edge $(S_i,S)$ in its path has to leave the
cluster $S_i$ and thus create a load of 1 in $\out_{S_i}$. Conversely, OPT has
congestion $C_{\mathrm{opt}}(d)$, i.e., a load of at most
$\out_{S_i}(S_i)\cdot C_{\mathrm{opt}}(d)$ on $\out_{S_i}$. Therefore the total demand
that has to be routed for $\RouteArg[1pt]{}{w_{S_i}}{w_S}$ is at most
$\out_{S_i}(S_i)\cdot C_{\mathrm{opt}}(d)=w_S(S_i)\cdot C_{\mathrm{opt}}(d)$.
An analogous argument holds for sending from $w_S$ to $w_{S_i}$.

Now, we define a CTS for every cluster $S$ that concurrently routes
$\RouteArg[2pt]{w_S(S_i)}{w_{S_i}}{w_S}$ for all child-clusters with small
congestion. If these schemes have congestion at most $C$ then the overall
competitive ratio of the compact oblivious routing scheme is $\height(T)C$ as
an edge is contained in at most $\height(T)$ many clusters.
Hence, we can restate our goal as follows. For every cluster $S$ find
\begin{quote}
\bold{Mixing CTS}\hfill\\
A CTS that routes
$\RouteArg[2pt]{w_{S}(S_i)}{w_{S_i}}{w_S}$ for each child $S_i$, with congestion $\tildeO(1)$.

\medskip
\bold{Unmixing CTS}\hfill\\
A CTS that routes
$\RouteArg[2pt]{w_{S}(S_i)}{w_S}{w_{S_i}}$ for each child $S_i$, with congestion $\tildeO(1)$.
\end{quote}
Here, we have to think about the encoding of the commodities, i.e., the indices of child clusters $S_i$. 
For our oblivious routing scheme we relabel the vertices so that the new name
of a vertex $v$ encodes the path from the root to the leaf $\{v\}$ in the
decomposition tree. Then when we are given a packet with a source and a target
node we can determine the path in the tree. For routing along an edge
$(S_i,S_{i+1})$ of this path we extract the name of the child cluster and use this as commodity for the CTS. Furthermore, we will fix a specific name for each child cluster, incorporating a little bit of information for the CTS. (As in the scheme by Räcke and Schmid~\cite{raecke2018compact}, the name will be the index in the list of child clusters, sorted by size.)

\subsection{Constructing Transformation Schemes}
\label{sec:theproblem} In this
section we give an overview of the steps for
constructing transformation schemes that for some cluster $S$ route
$\RouteArg[2pt]{w_{S}(S_i)}{w_{S_i}}{w_S}$ and
$\RouteArg[2pt]{w_{S}(S_i)}{w_S}{w_{S_i}}$ with small congestion. For this we
use simplified versions of the main lemmata that are proven in the technical
analysis in Section~\ref{sec:analysis}. We mark these simplified version
with a \enquote{$\,'\,$}, so Lemma
~3$'$ would be the simplified version
of Lemma~3 in Section~\ref{sec:analysis}.

\paragraph*{Single-commodity flows}
The first lemma that we show is how to construct a transformation scheme
from a given flow, to route between \emph{integral} distributions.

\begin{manuallemma}{\ref{lem:sctopath}\strich}[Single-commodity flow
routing] \label{lem:sctopath2} Let $f$ denote a flow with congestion at most
$\mathcal{O}(\operatorname{poly}(nW))$, and $\mu, \mu'$ integral
distributions with $\mu'-\mu=\operatorname{bal}_f$.\footnote{We remark that due
  to flow conservation, $\mu(V)=\mu'(V)$ holds.} Then there exists a
compact, deterministic transformation scheme that routes
$\RouteArg[3pt]{\mu(V)}{\mu}{\mu'}$ with congestion $\mathcal{O}(\cong(f))$
and has $N_v:=\mu(v)$ valid path-ids at node $v$.
\end{manuallemma}

This means that if we are given a flow then we can construct a transformation
scheme that allows us to send packets from sources (outgoing net-flow) to
targets (incoming net-flow). Note that there is no guarantee which target
a packet will be sent to if the flow contains several targets.

\paragraph*{Product multicommodity flow}
The second step of our approach is to obtain a concurrent transformation scheme
that routes a \emph{product multicommodity flow} (PMCF). Suppose that we are given a
weight function $\we:V\rightarrow\mathbb{N}$ on the vertices of the graph. We
associate a multicommodity flow problem with this weight function by defining a
demand $d(u,v)=\we(u)\we(v)/\we(V)$ for any pair $(u,v)$ of vertices. One can view
this demand as each vertex $u$ generating a flow of $\we(u)$ and distributing it
according to $\bwe$. Suppose that we can solve this multicommodity flow
problem with some congestion $C=\mathcal{O}(\mathrm{poly}(nW))$. We show that
we can then obtain a CTS that routes a solution to the PMCF.

\begin{manuallemma}{\ref{lem:krvtopath}\strich}[PMCF-routing] \label{lem:krvtopath2}
Given a graph $G$ together with a weight function $w:V\rightarrow\mathbb{N}$
on the vertices there exists a compact, deterministic CTS that routes
$\RouteArg[3pt]{\we(u)}{\textbf{1}_u}{\we}$ for each $u\in V$ with approximation
$1+\mathcal{O}(n^{-1})$ and congestion $\tildeO(C)$.
\end{manuallemma}
We obtain this result by making use of the KRV-framework~\cite{khandekar2009graph}. One way to
view this framework is that it tries to embed an expander into a graph by
solving a small number of single-commodity maximum flow problems. Each maximum
flow solution gives rise to a matching. One can then route to a random vertex by
following the \enquote{\emph{matching random walk}}, i.e., in the $i$-th step
the packet takes the (embedded) matching edge with probability $1/2$.

We proceed slightly differently. Instead of decomposing the flow into matchings
and then route along the matchings (which seems difficult to do with small
routing tables) we simply use Lemma~\ref{lem:sctopath2} to route along the flow.
This means in the $i$-th step we stay with probability 1/2 or we route the
packet along the flow to some target of the flow. More concretely, assume that the
KRV-scheme uses a flow $f$ between sources $S:=\{v\in V\mid \bal_f(v)<0\}$
and targets $T:=\{v\in V\mid \bal_f(v)>0\}$
in the $i$-th step. Then we construct two transformation schemes.
Let
\begin{equation*}
\mu(v):=\left\{\begin{array}{ll}
                 -\bal_f(v)&v\in S\\
                 0&\text{otw.}\\
              \end{array}\right.
              \quad\text{ and }\quad
\mu'(v):=\left\{\begin{array}{ll}
                 \bal_f(v)&v\in T\\
                 0&\text{otw.}\\
              \end{array}\right.              
\end{equation*}
We use Lemma~\ref{lem:sctopath2} to construct a transformation scheme
that routes $\mu\rightarrow \mu'$ and one that routes $\mu'\rightarrow\mu$.
These then allow us to distribute packets in the described way. The guarantees
of KRV still hold for this slightly modified scheme, which means that after
performing a polylogarithmic number of such steps a packet is at a
random location.

The above process and the transformation schemes for the individual iterations
can be combined into a single concurrent transformation scheme. The id-sets $M_v$
for this scheme contain bitstrings that encode for every iteration: (a) the
id to be used in the transformation scheme for this iteration and (b) a bit
that indicates whether to route along the flow or to stay.
Note that the CTS is deterministic, i.e., after choosing the id the
packet follows a fixed path in the network.

\begin{figure}[t]
\begin{center}
\def\svgwidth{12.7cm}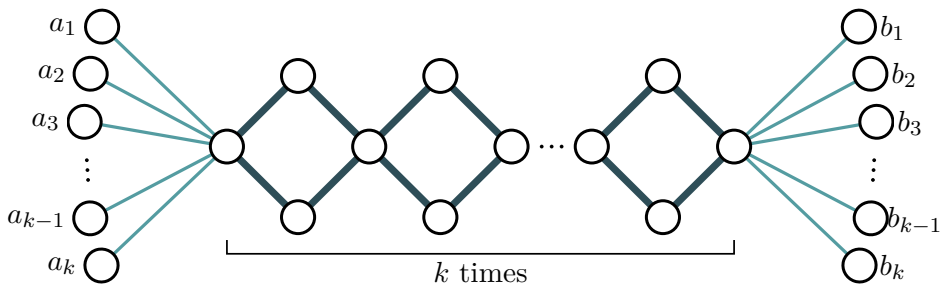
\end{center}
\caption{Routing multiple paths over a single edge. Each node $a_i$ sends a packet to $b_i$. A single path needs $k$ bits of information to encode, there are $k$ paths and $n=\mathcal{O}(k)$ nodes, so on average each node needs to store $k^2 / n=\Omega(n)$ bits if an arbitrary set of paths is chosen, which is too much. The same holds for paths chosen uniformly at random.}
\label{fig:smalloverbig}
\end{figure}

The result by R\"acke and Schmid~\cite{raecke2018compact} also required a sub-routine for
routing a product multicommodity flow. They used a randomized rounding approach
on the multicommodity flow solution of the PMCF-instance, and crucially
exploited the fact that the number of routing paths going through an edge in
such a solution is $\tildeO(C)$. As illustrated in Figure~\ref{fig:smalloverbig},
we cannot do this in our scenario as
the number of paths going through an edge might be as large as $\tildeO(CW)$.
Storing these paths would require large routing tables.

\paragraph*{Routing arbitrary demands}
The PMCF-scheme described above routes $\RouteArg[2pt]{w(u)}{\one_u}{\we}$. If we
choose $\we:=w_S$ for some cluster $S$ then this scheme gives us the first part
of our goal: we can concurrently route
$\RouteArg[2pt]{w_{S}(S_i)}{w_{S_i}}{w_S}$ with small congestion. To see this,
observe that if we consider the demands in the PMCF-scheme for all commodities
$u\in S_i$ combined, this is $\sum_uw_S(u)\one_u = w_S\upharpoonright S_i=\operatorname{out}_{S_i}$.
We can go from $w_{S_i}$ to $\operatorname{out}_{S_i}$ \emph{within $S_i$} using Lemma~\ref{lem:sctopath2}, and then the PMCF-scheme distributes it according to $w_S$. 

Hence, by simply merging commodities
$u\in S_i$ into one we obtain our desired CTS.

However, for our oblivious routing scheme we also need to be able to route
commodities $\RouteArg[3pt]{w_{S}(S_i)}{w_S}{w_{S_i}}$. This turns out to
be much more involved. Note that we cannot simply \enquote{route in reverse}
because a transformation scheme is inherently directed.

We do not directly construct a transformation scheme that routes
$\RouteArg[2pt]{w_{S}(S_i)}{w_S}{w_{S_i}}$ but we first embed an auxiliary
graph into the cluster $S$ (via a transformation scheme). This auxiliary graph
is directed and must have small degree for the embedding to be compact.

\begin{manuallemma}{\ref{lem:krmtoperm}\strich}[general graph embedding]\label{lem:krmtoperm2}
Let $G'=(S, A, d)$ denote a weighted, directed graph, where
the total weight of incoming and outgoing edges of a node $v$ is at most
$w_S(v)$, and $\deg_{G'}(v)\in\tildeO(\deg(v))$. Then
there is a compact CTS that routes $\RouteArg[2.5pt]{d(u,v)}{\one_u}{\one_v}$
for each $(u,v)\in A$ with congestion $\tildeO(C)$.
\end{manuallemma}

This lemma is the main technical contribution of our work. A rough outline of
the approach is as follows. The first observation is that one could combine the
result for the PMCF-routing with Valiant's trick~\cite{VB81,Va82} of routing to random
intermediate destinations. Suppose that we want to route from $x$ to $y$. Then
we first apply the PMCF-scheme of Lemma~\ref{lem:krvtopath2}; this brings us to
a node $z$ chosen according to $w_S$. At $z$ we choose a path id
$\mathit{id}_y$ that brings us to $y$, i.e., if we apply the
transformation scheme starting from node $z$ with id $\mathit{id}_y$ the packet
is delivered to $y$. We choose $\mathit{id}_y$ uniformly at random from all
path ids that will deliver the packet to $y$. This applies Valiant's trick
and the standard analysis shows that the congestion of this approach will be
$\tildeO(C)$.

However, implementing this approach with small routing tables is problematic.
The node $z$ could store a table of path ids which can be used for routing to
$y$ but this is clearly not compact.

If all edges have weight $1$, we can apply a suitable randomized rounding to
the above path generation method. Then the number of paths that go from $x$ to $y$ are just
$\tildeO(\deg(x))$. This allows the node $x$ to store the necessary information
for every path. In the weighted setting, however, the randomized rounding
approach leads to $\tildeO(W\cdot\deg(x))$ many paths. The resulting tables
would not be compact.

Instead we proceed as follows. We say a path $p$ is a \emph{class $\ell$ path}
if the smallest capacity edge of $p$ is from class $\ell$. (Recall that we assume all capacities to be powers of two, and that a class $\ell$ edge is one with capacity $2^\ell$.)

A pair $(x,y)$ is
from class $\ell$ if $\ell$ is the most frequent class that occurs when
generating $x$-$y$-paths by the above process.

In a first step we change the path generation process to only use class $\ell$
paths for a class $\ell$ pair. This only increases the congestion by a factor
of $\Nclass$ (the number of classes). For a randomized rounding approach to
guarantee a good congestion we need to spread the traffic between a class
$\ell$ pair $(x,y)$ over roughly $k:=d(x,y)/2^\ell$ many class $\ell$ paths.

For this we split the packet into $k$ different parts, each with a different path. However, we cannot store information for each such path in $x$ directly, as $k$ may be large. Instead, we identify $k$ many other nodes, each of which we use to store the information for just a single path. Of course, we cannot simply pick any nodes in the cluster ― they have to be reachable from $x$ with low congestion.

As it turns out, the set of $k$ class $l$ paths from $x$ to $y$ already contains an appropriate choice. Each class $l$ path contains a class $l$ edge, and we pick one of its two adjacent nodes as helper node. Now observe that each path transports $2^\ell$ flow, so there can only be a small number of paths using that edge, because we have low congestion. That means that we can use $\tildeO(1)$ space in the helper node, for each path that uses it.

Now that we have found $k$ helper nodes that can store our routing information, the packets still have to reach these nodes, to pick up that information. So now we send a single-commodity flow from \emph{all} source nodes $x'$ of class $\ell$ pairs to their helper nodes, and then back. We set up a TS for both directions of the flow, using Lemma \ref{lem:sctopath}. 

Note that this does not guarantee that a packet from source node $x$ reaches \enquote{its} helper node, but this is not required ― it only needs to reach \emph{a} node in which to store its routing information. Similarly, the packet may not get back to $x$, but end up at a different $x'$.

We have the same packet distribution as before, meaning every $x'$ has the same number of packets, but possibly different ones. However, each packet has picked up $\tildeO(1)$ of information while passing its helper node.

Suppose that all packets now in $x$ have a target $y'$ s.t.\ $(x,y')$ is a class $\ell'\ge\ell$ pair. Then we can pick a single path for each of these packets, and store the information for that path in the helper node. (Recall that paths are generated by the PMCF-scheme of Lemma \ref{lem:krvtopath2}, so we only need to store their path ids.)

If that is not the case we split the packet further, by applying the scheme recursively.

\paragraph*{Hypercube embedding} The previous lemma tells us that we can embed
graphs with low degree. More precisely, we can embed any graph $G'=(S,A,d)$ which fulfills the following properties.
\begin{enumerate}[(1)]
\item The degree $\deg_{G'}(v)$ at a vertex is polylogarithmic.
\item The capacity of incoming edges and the capacity of outgoing edges at a
vertex is roughly equal to $w_S(v)$.
\end{enumerate}
Now, in order to be able to construct an unmixing CTS, i.e., route
$\RouteArg[2pt]{w_{S}(S_i)}{w_S}{w_{S_i}}$ for a cluster $S$, we find some graph with the above properties where an unmixing CTS is easy to implement, and then embed that graph into $G$. In particular, we want a $G'$ which has one additional property.
\begin{enumerate}[(1)]\setcounter{enumi}{2}
\item There is a suitable numbering of the child clusters of $S$ for which
there exists a CTS for $G'$ that routes
$\RouteArg[3pt]{w_{S}(S_i)}{w_S}{w_{S_i}}$ for each $S_i$ with small congestion and commodity $S_i$ encoded as integer $i$.
\end{enumerate}

The result by Räcke and Schmidt~\cite{raecke2018compact} used a hypercube, where each node $v$ received $w_S(v)$ hypercube ids. For weighted edges this would violate property (1).

Instead, we essentially embed several (disconnected) hypercubes, one for each class $\ell$. A node $v$ then receives roughly $w^{(\ell)}_S(v)$ hypercube ids, at most one for each class $\ell$ edge adjacent to it.
The existence of a good CTS scheme for $G'$ then
follows from classical results about online routing on the hypercube~\cite{VB81}.

%% file: figures/weight_cluster.pdf_tex
\begingroup%
  \makeatletter%
  \providecommand\color[2][]{%
    \errmessage{(Inkscape) Color is used for the text in Inkscape, but the package 'color.sty' is not loaded}%
    \renewcommand\color[2][]{}%
  }%
  \providecommand\transparent[1]{%
    \errmessage{(Inkscape) Transparency is used (non-zero) for the text in Inkscape, but the package 'transparent.sty' is not loaded}%
    \renewcommand\transparent[1]{}%
  }%
  \providecommand\rotatebox[2]{#2}%
  \newcommand*\fsize{\dimexpr\f@size pt\relax}%
  \newcommand*\lineheight[1]{\fontsize{\fsize}{#1\fsize}\selectfont}%
  \ifx\svgwidth\undefined%
    \setlength{\unitlength}{198.42519685bp}%
    \ifx\svgscale\undefined%
      \relax%
    \else%
      \setlength{\unitlength}{\unitlength * \real{\svgscale}}%
    \fi%
  \else%
    \setlength{\unitlength}{\svgwidth}%
  \fi%
  \global\let\svgwidth\undefined%
  \global\let\svgscale\undefined%
  \makeatother%
  \begin{picture}(1,1)%
    \lineheight{1}%
    \setlength\tabcolsep{0pt}%
    \put(0,0){\includegraphics[width=\unitlength,page=1]{weight_cluster.pdf}}%
  \end{picture}%
\endgroup%

%% file: figures/weight_border.pdf_tex
\begingroup%
  \makeatletter%
  \providecommand\color[2][]{%
    \errmessage{(Inkscape) Color is used for the text in Inkscape, but the package 'color.sty' is not loaded}%
    \renewcommand\color[2][]{}%
  }%
  \providecommand\transparent[1]{%
    \errmessage{(Inkscape) Transparency is used (non-zero) for the text in Inkscape, but the package 'transparent.sty' is not loaded}%
    \renewcommand\transparent[1]{}%
  }%
  \providecommand\rotatebox[2]{#2}%
  \newcommand*\fsize{\dimexpr\f@size pt\relax}%
  \newcommand*\lineheight[1]{\fontsize{\fsize}{#1\fsize}\selectfont}%
  \ifx\svgwidth\undefined%
    \setlength{\unitlength}{198.42519685bp}%
    \ifx\svgscale\undefined%
      \relax%
    \else%
      \setlength{\unitlength}{\unitlength * \real{\svgscale}}%
    \fi%
  \else%
    \setlength{\unitlength}{\svgwidth}%
  \fi%
  \global\let\svgwidth\undefined%
  \global\let\svgscale\undefined%
  \makeatother%
  \begin{picture}(1,1)%
    \lineheight{1}%
    \setlength\tabcolsep{0pt}%
    \put(0,0){\includegraphics[width=\unitlength,page=1]{weight_border.pdf}}%
  \end{picture}%
\endgroup%

%% file: figures/weight_child.pdf_tex
\begingroup%
  \makeatletter%
  \providecommand\color[2][]{%
    \errmessage{(Inkscape) Color is used for the text in Inkscape, but the package 'color.sty' is not loaded}%
    \renewcommand\color[2][]{}%
  }%
  \providecommand\transparent[1]{%
    \errmessage{(Inkscape) Transparency is used (non-zero) for the text in Inkscape, but the package 'transparent.sty' is not loaded}%
    \renewcommand\transparent[1]{}%
  }%
  \providecommand\rotatebox[2]{#2}%
  \newcommand*\fsize{\dimexpr\f@size pt\relax}%
  \newcommand*\lineheight[1]{\fontsize{\fsize}{#1\fsize}\selectfont}%
  \ifx\svgwidth\undefined%
    \setlength{\unitlength}{198.42519685bp}%
    \ifx\svgscale\undefined%
      \relax%
    \else%
      \setlength{\unitlength}{\unitlength * \real{\svgscale}}%
    \fi%
  \else%
    \setlength{\unitlength}{\svgwidth}%
  \fi%
  \global\let\svgwidth\undefined%
  \global\let\svgscale\undefined%
  \makeatother%
  \begin{picture}(1,1)%
    \lineheight{1}%
    \setlength\tabcolsep{0pt}%
    \put(0,0){\includegraphics[width=\unitlength,page=1]{weight_child.pdf}}%
  \end{picture}%
\endgroup%

%% file: figures/smalloverbig.pdf_tex
\begingroup%
  \makeatletter%
  \providecommand\color[2][]{%
    \errmessage{(Inkscape) Color is used for the text in Inkscape, but the package 'color.sty' is not loaded}%
    \renewcommand\color[2][]{}%
  }%
  \providecommand\transparent[1]{%
    \errmessage{(Inkscape) Transparency is used (non-zero) for the text in Inkscape, but the package 'transparent.sty' is not loaded}%
    \renewcommand\transparent[1]{}%
  }%
  \providecommand\rotatebox[2]{#2}%
  \newcommand*\fsize{\dimexpr\f@size pt\relax}%
  \newcommand*\lineheight[1]{\fontsize{\fsize}{#1\fsize}\selectfont}%
  \ifx\svgwidth\undefined%
    \setlength{\unitlength}{416.69291339bp}%
    \ifx\svgscale\undefined%
      \relax%
    \else%
      \setlength{\unitlength}{\unitlength * \real{\svgscale}}%
    \fi%
  \else%
    \setlength{\unitlength}{\svgwidth}%
  \fi%
  \global\let\svgwidth\undefined%
  \global\let\svgscale\undefined%
  \makeatother%
  \begin{picture}(1,0.30612245)%
    \lineheight{1}%
    \setlength\tabcolsep{0pt}%
    \put(0,0){\includegraphics[width=\unitlength,page=1]{smalloverbig.pdf}}%
    \put(0.08380109,0.26976397){\color[rgb]{0,0,0}\makebox(0,0)[rt]{\lineheight{1.25}\smash{\begin{tabular}[t]{r}$a_1$\end{tabular}}}}%
    \put(0.94167717,0.26462913){\color[rgb]{0,0,0}\makebox(0,0)[rt]{\lineheight{1.25}\smash{\begin{tabular}[t]{r}$b_1$\end{tabular}}}}%
    \put(0.07192352,0.22057005){\color[rgb]{0,0,0}\makebox(0,0)[rt]{\lineheight{1.25}\smash{\begin{tabular}[t]{r}$a_2$\end{tabular}}}}%
    \put(0.06566913,0.1703923){\color[rgb]{0,0,0}\makebox(0,0)[rt]{\lineheight{1.25}\smash{\begin{tabular}[t]{r}$a_3$\end{tabular}}}}%
    \put(0.07153058,0.06959627){\color[rgb]{0,0,0}\makebox(0,0)[rt]{\lineheight{1.25}\smash{\begin{tabular}[t]{r}$a_{k-1}$\end{tabular}}}}%
    \put(0.08315699,0.02034339){\color[rgb]{0,0,0}\makebox(0,0)[rt]{\lineheight{1.25}\smash{\begin{tabular}[t]{r}$a_k$\end{tabular}}}}%
    \put(0.95355474,0.21543521){\color[rgb]{0,0,0}\makebox(0,0)[rt]{\lineheight{1.25}\smash{\begin{tabular}[t]{r}$b_2$\end{tabular}}}}%
    \put(0.95980908,0.16525746){\color[rgb]{0,0,0}\makebox(0,0)[rt]{\lineheight{1.25}\smash{\begin{tabular}[t]{r}$b_3$\end{tabular}}}}%
    \put(0.98013503,0.06446164){\color[rgb]{0,0,0}\makebox(0,0)[rt]{\lineheight{1.25}\smash{\begin{tabular}[t]{r}$b_{k-1}$\end{tabular}}}}%
    \put(0.94232115,0.01520876){\color[rgb]{0,0,0}\makebox(0,0)[rt]{\lineheight{1.25}\smash{\begin{tabular}[t]{r}$b_k$\end{tabular}}}}%
    \put(0,0){\includegraphics[width=\unitlength,page=2]{smalloverbig.pdf}}%
    \put(0.50156532,0.00760697){\color[rgb]{0,0,0}\makebox(0,0)[t]{\lineheight{1.25}\smash{\begin{tabular}[t]{c}$k$ times\end{tabular}}}}%
  \end{picture}%
\endgroup%

%% file: appendix.tex
\section{Detailed Analysis}
\label{sec:analysis}

%
%
%
%

In this section we provide the details for constructing a mixing and an
unmixing CTS for every cluster. This will then give the oblivious routing scheme.

Most of our lemmata related to the PMCF work for arbitrary weights $c$ where the corresponding PMCF can be solved with congestion $\Co\in\mathcal{O}(\operatorname{poly}(nW))$. We prove those without making further assumptions.

Note also that we can restrict the transformation schemes to route within a cluster: If a lemma is proven for arbitrary weights $c$, then it also works for weights $w_S$ within the subgraph $G[S]$.

\subsection{Constructing a mixing CTS}
\label{sec:mixing}

\paragraph*{Single-commodity flows}
To get started, we show how to use a single-commodity flow to construct a transformation scheme. The key idea is that we decompose the flow into paths with unique identifiers and store intervals for each edge, to encode the outgoing paths.

\begin{lemma}[Flow routing] \label{lem:sctopath}
Let $f$ denote a flow with $\cong(f)\in\mathcal{O}(\operatorname{poly}(nW))$, and $\mu, \mu'$ integral distributions with $\mu'-\mu=\operatorname{bal}_f$. Then there exists a deterministic TS that routes $\RouteArg[3pt]{\mu(V)}{\mu}{\mu'}$ with congestion $\mathcal{O}(\cong(f))$.

The routing table of node $v$ has size $\mathcal{O}(\deg(v)\log (nW))$, and packet headers have length $\mathcal{O}(\log (nW))$. At a node $v$ there are $N_v:=\mu(v)$ valid path ids.
\end{lemma}
\begin{proof}
First, we transform $f$ to be integral and acyclic. 

Let $F:=\lceil\operatorname{cong}(f)\rceil$. We consider the single-commodity
flow problem where we add a source $s$, a sink $t$ and edges $(s,v)$, $(v,t)$
for $v\in V$ with respective capacities $\mu(v)$ and $\mu'(v)$. We retain the
other edges, scaling their capacities by $F$. All capacities
are integral and $f$ is a solution with flow value $\mu(V)$, so there is also an
integral, acyclic solution $f'$, which by construction has
$\mu'-\mu=\operatorname{bal}_{f'}$ and $\operatorname{cong}(f')\le F$.

For each source $v\in V$ (i.e., a node that has $k:=-\!\bal_{f'}(v)>0$) we put
$k$ tokens into $v$. These have labels $a+1,a+2,...,a+k$ where $a$ is chosen
s.t.\ the labels are disjoint for all nodes. We store $a$ and $k$ in $v$, which
takes $\mathcal{O}(\log (nW))$ space.

Now we repeatedly move those tokens according to $f'$, iteratively constructing the
TS in the process. Each token represents a unit of flow.

As $f'$ is acyclic, we iterate over the nodes based on the topological ordering
given by $f'$. Let $u$ denote the current node. We assume the invariant that no
previous node contains any tokens, which is true in the beginning and will hold
inductively for all other iterations.

Therefore, $u$ has tokens precisely equal to the flow over incoming edges of
$u$ (plus $-\!\bal_{f'}(u)$ if $u$ is a source), which, due to flow conservation,
is the same as the flow over outgoing edges (plus $\bal_{f'}(u)$ if $u$ is a sink). We
distribute the tokens by sorting them and assign each outgoing edge $(u,v)$ an
amount of consecutive tokens, according to its flow $f'(u,v)$. These tokens are
sent over that edge. Exactly $\max\{\bal_{f'}(v), 0\}$ tokens
remain, which we remove from the graph. Hence we deleted all tokens from $u$
and added tokens only to its successors, so our invariant still holds.

After the last iteration, all nodes are empty, so each token has been routed.
To construct a TS, we encode the path of all tokens, by storing an
\emph{interval} $I(u,v)$ of ids for each edge $(u,v)$, which contains the
tokens which were routed over that edge. Note that the interval may be larger
than the number of tokens that use $e$; all tokens that traverse $u$ and are
inside $I(u,v)$ use edge $e=(u,v)$, but the interval may contain tokens that do
not traverse $u$. In total we just need to store two ids of length
$\mathcal{O}(\log (nW))$ per incident edge.
(Recall that the
maximum load over any edge is at most
$FW\in\mathcal{O}(\operatorname{poly}(nW))$.)

As we send exactly $|f'(e)|$ tokens over an edge $e$ we get the same congestion
as $f'$. 

We want to construct a deterministic TS, so the input label at a node $v$ contains a path id $\ell\in\{1,...,\mu(v)\}$. As $v$ stores the offset $a$ from above, we can map the path id to $a+\ell$, one of the tokens starting at $v$. Afterwards, the routing algorithm simply needs to check which interval contains the token, to simulate their movement. This is deterministic.
\end{proof}

We use Lemma \ref{lem:sctopath} typically to route between distributions which are ‘close’ to our weights $\we$. As we can solve the PMCF with low congestion, this will have low congestion as well. The following lemma encapsulates that argument.

\begin{lemma}\label{lem:routingscf}
Let $\muin,\muout$ denote distributions with $\muin,\muout\le\we$ and $\muin(V)=\muout(V)$. Then there is a flow $f$ with $\operatorname{bal}_f=\muout-\muin$ and $\cong(f)\le 2C$.
\end{lemma}
\begin{proof}
We construct $f$ based on the PMCF, using Valiant's trick.

Let $(f_{u,v})_{u,v\in V}$ denote a solution to the PMCF with congestion at most $C$, where $f_{u,v}$ is a unit $u$-$v$-flow, i.e., it sends one unit of flow from $u$ to $v$ and has to be scaled by the corresponding demand.

Intuitively, we route from $u$ to $v$ by sending a packet to an intermediate node $w$, picked randomly weighted by $\we$. The route itself uses the flows $f_{u,w}$ and $f_{w,v}$, which we can concatenate by adding them. Hence, 
\[f:=\sum_{u,v\in V} f_{u,v}\cdot\muin(u)\bwe(v) + \sum_{u,v\in V} f_{u,v}\cdot\bwe(u)\muout(v)\]
As both $\muin(u)\bwe(v)$ and $\bwe(u)\muout(v)$ are at most $\we(u)\we(v)/\we(V)$, the demand of the PMCF, we have $\cong(f)\le 2C$. The flow to and from the intermediate nodes cancels, so a node $v$ sends $\muin(v)$ packets and receives $\muout(v)$, yielding $\operatorname{bal}_f=\muout-\muin$.
\end{proof}

\paragraph*{Routing the PMCF}
Our first goal is to create a transformation scheme for the PMCF, for which we use the technique of cut-matching games introduced in \cite{khandekar2009graph}. We will not discuss it in detail, but instead encapsulate the properties of interest and refer to the original proofs. We need modify the technique slightly
\footnote{In particular, we use a bijection instead of a perfect matching, allow the matching player to choose subsets which will be shuffled randomly, and require a stronger bound on the error.}
, so for those parts we briefly show how the proofs can be adapted.

Consider the following game. We are given a finite set of nodes $V'$, with $|V'|$ even. There are $N\in\mathcal{O}(\log^2 |V'|)$ rounds and two players. In round $k$,
\begin{itemize}
\item Player 1 (the “cut player”) chooses a partition $A_1\dotcup A_2=V'$ with $|A_1|=|A_2|$,
\item Player 2 (the “matching player”) chooses a bijection $M:V'\rightarrow V'$ \emph{respecting the partition}, i.e., it maps $A_1$ to $A_2$ and vice versa, and
\item Player 2 chooses a partition $B$ of $V'$.
\end{itemize}

At the end of the game, we define a random walk on $V'$ consisting of $N$ steps. In step $k$, a packet
\begin{itemize}
\item moves from node $v$ to either $v$ or $M(v)$ with probability $\frac{1}{2}$, and then
\item moves from the resulting node $v'$ to a node in $B_{v'}$ uniformly at random, where $v'\in B_{v'}\in B$ is the group of $v'$ in the partition $B$.
\end{itemize}
The game is won by Player 1 if this random walk is \emph{mixing}, i.e., for any $u,v\in V'$ it moves from $u$ to $v$ with probability between $1/|V'|\pm\varepsilon$, where $\varepsilon=1/|V'|^2$.

\begin{lemma}[KRV]\label{lem:krv}
Player 1 has a winning strategy.
\end{lemma}
\begin{proof}
See \cite[Section 3.1]{khandekar2009graph}. The proofs have to be changed slightly to work here, which we will now do, using the original notation.

While the original result uses perfect matchings instead of bijections for $M$, it extends directly by slightly modifying the proof of Lemmata 3.1 and 3.3. In particular, we can decompose $M$ into two perfect matchings $m_1=M\upharpoonright A_1$ and $m_2=M\upharpoonright A_2$. We write
\[\psi_M(t):=\sum_{(i,j)\in M}\Norm{\frac{P_i(t)+P_j(t)}{2}-\frac{\one}{n}}^2\]
for the resulting potential after adding $M$ to the random walk. Note that a perfect matching, such as $m_1$, changes the potential to $2\psi_{m_1}(t)$, as we need to count each node twice. We now have $\psi_M(t)=(2\psi_{m_1}(t)+2\psi_{m_2}(t))/2$, and the original result guarantees the reduction for both halves.

Moving from $v'$ to a node in $B_{v'}$ corresponds to averaging the probability vectors of the nodes in that cell, which does not increase the potential function: For any set of vectors $v_1,...,v_k$ the Cauchy-Schwarz inequality yields
\[k\,\Big\lVert\frac{1}{k}\sum_{i}v_i\Big\rVert^2
\le \frac{1}{k}\Big(\sum_{i}1\cdot \norm{v_i}\Big)^2
\le \frac{1}{k}\Big(\sum_{i}1^2\cdot\sum_{i}\norm{v_i}^2\Big)=\sum_{i}\norm{v_i}^2\]

The original result just guarantees $\varepsilon\le 1/2|V'|$, but we can increase the number of iterations by a constant factor.
\end{proof}

These random walks can be made deterministic, by storing whether to switch sides at each step, provided that we can send packets along our chosen bijection $M$. Additionally, while we need a lot of nodes in $V'$ to match the weights~$\we$, a single node $v$ can simulate many in $V'$ with only little bookkeeping. Here we use that \enquote{mixing} nodes of $V'$ at each step is not problematic, a notion made precise by choosing the appropriate partition $B$. The bijections $M$ will be stored implicitly using single-commodity flows.

In total we get short descriptions of the possible paths taken by the random walk, enabling us to circumvent one of the key problems arising in weighted graphs ― the inability to store paths directly within the graph due to too many paths going over a single edge.

\begin{lemma}[PMCF transformation scheme] \label{lem:krvtopath}
There exists a deterministic CTS that routes $\RouteArg[3pt]{\we(v)}{\one_v}{\we}$ for each $v\in V$ with approximation $1+\mathcal{O}(n^{-1})$ and congestion $\mathcal{O}(\Co\log^2n)$.

The routing table of node $v$ has size $\mathcal{O}(\deg(v)\log^3(nW))$, while path ids and packet headers have length $\mathcal{O}(\log^3(nW))$.
\end{lemma}
\begin{proof}
We want to play the game described above Lemma \ref{lem:krv}, so we define a set of “virtual” nodes $V':=\{1,...,2n\we(V)\}$
\footnote{We would like to have $V'=\{1,...,\we(V)\}$, but we need to make sure that $|V'|$ is even and at least $n$.}
and choose an embedding $\varphi: V'\rightarrow V$ which assigns each node virtual nodes according to its weight, i.e.\ $|\varphi^{-1}(v)|=2n\we(v)$.

In each turn we have $A_1\cup A_2=V', |A_1|=|A_2|$ and can choose any bijection $M$ respecting that partition. We want to simulate the random walk, so we also need a way to send packets according to $M$. In other words, we need a CTS that routes $\Route{\one_{\varphi(v')}}{\one_{M(\varphi(v'))}}$ for each $v'\in V'$. It is difficult to construct such a CTS \emph{given} a specific $M$, so instead we build the transformation scheme first, and then define $M$ accordingly.

We construct two deterministic transformation schemes $\mathit{TS}_1$ and $\mathit{TS}_2$, routing from $A_1$ to $A_2$ and vice versa. Let us first consider $\mathit{TS}_1$.

A node $v$ sends out a packet for each virtual node in $A_1$ assigned to it, so $\mu_1(v):=|\varphi^{-1}(v)\cap A_1|$ in total. Analogously, it receives $\mu_2(v):=|\varphi^{-1}(v)\cap A_2|$ packets. Hence we want to route $\RouteArg[3pt]{n\we(V)}{\mu_1}{\mu_2}$.

We have $\mu_1,\mu_2\le\we$ and $\mu_1(V)=\mu_2(V)=n\we(V)$, so Lemma \ref{lem:routingscf} yields a flow $f$ with $\operatorname{bal}_f=\mu_2-\mu_1$ and $\cong(f)\le 2n\we(V)$. (We scale $\mu_1,\mu_2$ by $1/n$ and the resulting flow by $n$.)

Using $f$, we apply Lemma \ref{lem:sctopath} to get a deterministic transformation scheme $\mathit{TS}_1$ that routes $\RouteArg[3pt]{n\we(V)}{\mu_1}{\mu_2}$ with congestion $\mathcal{O}(n\Co)$. 

At each node $v$ there are now $\mu_1(v)$ path ids to route a packet using $\mathit{TS}_1$. The transformation scheme is deterministic, so each path id corresponds to a single target node. Consider giving $\mu_1(v)$ packets to each node $v$, one for each path id, and routing them accordingly. Then, $v$ will receive $\mu_2(v)$ packets. By mapping the outgoing packets to $\varphi^{-1}(v)\cap A_1$ and the incoming packets to $\varphi^{-1}(v)\cap A_2$ in some fashion, we get a bijection from $A_1$ to $A_2$.

We construct $\mathit{TS}_2$ in the same manner, and combine the two mappings to get the desired bijection $M$ from $V'$ to $V'$. 

It remains to be shown that we can indeed route this bijection efficiently, i.e., without encoding any arbitrary mappings between virtual nodes and path ids.

Instead, we will simply choose a \emph{random} path id for either $\mathit{TS}_1$ or $\mathit{TS}_2$, weighted such that each path id has the same probability. This corresponds to moving to a random virtual node assigned to $v$ in each iteration, i.e.\ choosing our partition as $B:=\{\varphi^{-1}(v):v\in V\}$ in each round.

To summarize, we route the random walk as follows. In each iteration $k=1,...,N$ we flip a fair coin to decide whether we move from node $v$ according to $M$ or not. If yes, we pick a number $i$ u.a.r.\ in $\{1, ..., 2n\we(v)\}$. The first $\mu_1(v)$ numbers stand for path ids of $\mathit{TS}_1$, the others for path ids of $\mathit{TS}_2$. Then we route using the given transformation scheme and path id.

As we want our transformation scheme to be deterministic, these random choices will not be made while routing the packet, but encoded in the path id. There is a small technical issue in that the set $\{1,...,2n\we(v)\}$ from which we sample $i$ depends on $v$, but needs to be encoded in a path id chosen u.a.r.\ from some \emph{fixed} range of integers. Instead, we will sample $i'\in\{1,...,2n^2\we(V)N\}$ u.a.r\ and set $i:=i'\pmod{2n\we(v)}$. (Recall that $N$ is the number of rounds.) So $i$ is not quite uniform, but the probabilities differ by at most a factor of $1+1/nN$ in each round, and $1+1/n$ in total.

Thus we can store all random choices for the $\mathcal{O}(\log^2n\we(V))=\mathcal{O}(\log^2(nW))$ iterations using $\mathcal{O}(\log^3(nW))$ bits. These are the path ids of our transformation scheme. As there is no randomness apart from the choices encoded in the path id, the transformation scheme is deterministic. The packet headers need to include the headers from Lemma \ref{lem:sctopath}, as well as our path ids; the length of the latter dominates. 

Recall that we want to have a CDS that routes $\RouteArg[3pt]{\we(v)}{\one_v}{\we}$ for each $v\in V$. Hence, in aggregate the input distribution is $\we$.

Let us now analyze the congestion. $\mathit{TS}_1$ routes $\RouteArg[3pt]{n\we(V)}{\mu_1}{\mu_2}$ and $\mathit{TS}_2$ does $\RouteArg[3pt]{n\we(V)}{\mu_2}{\mu_1}$, both with congestion $\mathcal{O}(n\Co)$. If we add them and scale the demand by $1/n$ we route $\RouteArg[3pt]{\we(V)}{\we}{\we}$ with congestion $\mathcal{O}(\Co)$. So the distribution of packets does not change in an iteration, except for the factor of $1+1/nN$ above, and the total congestion is $\mathcal{O}(\Co\log^2(nW))$.

Due to Lemma \ref{lem:krv}, the random walk moves to any virtual node with probability between $1/2n\we(V)\pm1/|V'|^2$. We have $|V'|^2\ge 2n^2\we(V)$, so scaling by the total amount of flow $\we(V)$ yields a value in $1/2n\pm1/2n^2$. A node $v\in V$ has $2n\we(v)$ virtual nodes, so it receives between $\we(v)(1\pm1/n)$ packets in the random walk, or between $\we(v)(1\pm2/n)$ in the actual transformation scheme. Hence the output distribution is $\we$, with an approximation of $1+\mathcal{O}(n^{-1})$.
\end{proof}

We remark that this CTS has input distribution $\one_v$ for commodity $v\in V$, which means that the source node of a packet already encodes its commodity.

\paragraph*{Mixing CTS}
To close out section we prove that we can implement the mixing step with the tools we have. To start, we need a small helper lemma.

\begin{lemma}[Routing distributions similar to $c$]\label{lem:routingsimilar}
Let $\muin,\muout$ denote integral distributions with $\muin,\muout\le\we$ and set $M:=\min\{\muin(V),\muout(V)\}$. Then there exists a deterministic TS that routes $\RouteArg[3pt]{M}{\muin(V)}{\muout(V)}$ with congestion $\mathcal{O}(\Co)$. The routing table of node $v$ has size $\mathcal{O}(\deg(v)\log (nW))$, while path ids and packet headers have length $\mathcal{O}(\log (nW))$.
\end{lemma}
\begin{proof}
From Lemma \ref{lem:routingscf} we get a flow $f$ with congestion $2\Co$, by scaling the distributions to $M\bmuin$ and $M\bmuout$. Scaling both $f$ and the distributions by $\alpha:=\muin(V)\muout(V)$, the latter are now integral again and we can apply Lemma \ref{lem:sctopath}. This routes $\RouteArg[3pt]{\alpha M}{\alpha M\bmuin}{\alpha M\bmuin}$ with congestion $\mathcal{O}(\alpha\Co)$, or equivalently $\RouteArg[3pt]{M}{\muin}{\muin}$ with congestion $\mathcal{O}(\Co)$.
\end{proof}

Now we fix a cluster $S$ with children $S_1,...,S_r$. As we will later change the numbering of children, it is important that the following lemma works for an arbitrary one.

\begin{lemma}[Mixing CTS]\label{lem:mixing}
There exists a CTS that routes $\RouteArg[3pt]{w_S(S_i)}{w_{S_i}}{w_S}$ for each $i=1,...,r$ with congestion $\mathcal{O}(\Co\log^2n)$ and approximation $1+\mathcal{O}(n^{-1})$. The routing table of node $v$ has size $\mathcal{O}(\deg(v)\log^3(nW))$, while path ids and packet headers have length $\mathcal{O}(\log^3(nW))$.
\end{lemma}
\begin{proof}
For each $S_i$ we route $\RouteArg[3pt]{w_S(S_i)}{w_{S_i}}{\operatorname{out}_{S_i}}$ \emph{within $S_i$} using Lemma \ref{lem:routingsimilar} for weights $c:=w_{S_i}$. This has congestion $\mathcal{O}(C)$, as $w_S(S_i)=\operatorname{out}(S_i)$. It uses space only within $S_i$, so $\mathcal{O}(\deg(v)\log (nW))$ per node $v\in S$.

In total, the packets are now in distribution $\sum_i\operatorname{out}_{S_i}=w_S$ and we apply Lemma \ref{lem:krvtopath} with $c:=w_S$. (Here we do not use that Lemma~\ref{lem:krvtopath} gives a deterministic CTS.) As all source nodes route to $w_S$ concurrently, we route $\RouteArg[3pt]{w_S(S_i)}{\operatorname{out}_{S_i}}{w_S}$ for each $i=1,...,r$ (with approximation $1+\mathcal{O}(n^{-1})$). This has congestion $\mathcal{O}(C\log^2n)$.

For the bounds on space per node and length of packet headers, the costs of the latter step dominate.
\end{proof}

As for Lemma~\ref{lem:krvtopath}, the source node of a packet already determines the commodity, so there is no need to specify an encoding for it.

\subsection{Constructing an Unmixing CTS}
\label{sec:unmixing}

\paragraph*{General Graph Embedding}
Up until now, we have not used that we have only edges of distinct classes. The next two lemmas concern randomized rounding, which we use to select a small number of paths from a flow without increasing congestion. This uses a probabilistic argument to prove existence, but the choice of paths is fixed and not subject to randomness.

Consider some flow $f$ sending $k$ packets from a source node $u$ to a node $v$ with congestion 1, where the flow involves only edges of weight 1. It is obvious that taking a single path with weight $k$ from that flow uniformly at random increases the congestion to $k$, while taking $k$ paths with weight one should intuitively work quite well, giving a congestion $1+\mathcal{O}(\log k)$. This intuition is correct, which we now prove formally.

We call a multi-set of $u$-$v$-paths a \emph{path system}. The \emph{class of a path} is the minimum class of its edges, and the \emph{class of a path system $P$} is the minimum class amongst its paths. To send a packet using a path system $P$ we choose a path uniformly at random. Therefore, if we have multiple path systems $P=\{P_1,...,P_k\}$ with demands $d=\{d_1,...,d_k\}$, then their total congestion is $\cong(P, d):=\cong(\{d_ip/|P_i|:i\in\{1,...,k\},p\in P_i\})$.

\begin{lemma}[Randomized rounding]\label{lem:randomizer2}
Let $P=\{P_1,...,P_k\}$ denote a set of path systems with demands $d$. Then there exists a set $P'=\{P_1',...,P_k'\}$, with $P'_i\subseteq P_i$ and $|P'_i|\le\lceil2^{-l}d_i\rceil$ for each $P_i$ of class $\ell$. The congestion is $\cong(P',d)\in\mathcal{O}(\cong(P,d)+\log n)$.
\end{lemma}
\begin{proof}
We use the probabilistic method, so we will choose $P_i'$ by picking the appropriate number of paths from $P_i$ randomly, and then show that the congestion is low enough with positive probability. The latter part uses the following bound:

\begin{lemma}[adapted from \mbox{\cite[lemma 10]{raecke2018compact}}]\label{lem:randomizer}
Let $X_1,...,X_n$ denote a set of negatively correlated random variables taking values in $[0,1]$. Let $X$ denote their sum, and let $\delta\ge\mathbb{E}(X)$. Then $\operatorname{Pr}(X\ge\mathbb{E}(X)+\delta)\le e^{-\delta/3}$.
\end{lemma}

First we consider the congestion on a particular edge $e$. For path system $P_i$ of class $\ell$ we sample $N_i:=\lceil2^{-l}d_i\rceil$ paths (with replacement), so we define random variables $X_{i,p,j}$ for each $p\in P_i, j\in\{1,...,N_i\}$ as the congestion induced on $e$. More precisely, if $e\in p$ and $p$ is picked as the $j$-th path, then $X_{i,p,j}=d_i/N_iw(e)$, else $X_{i,p,j}=0$. Note that an $e\in p\in P_i$ has class at least $\ell$, so $X_{i,p,j}\le 1$.

Finally, $X_e:=\sum X_{i,p,j}$ is the total congestion of edge $e$. Of course, each path still has the same probability, so $\mathbb{E}(X_e)$ equals the original congestion on $e$ w.r.t.\ $P$ and demands $d$. Choosing $\delta:=6\ln m+\mathbb{E}(X_e)$ for Lemma \ref{lem:randomizer}, we get $\operatorname{Pr}(X_e\ge2\mathbb{E}(X_e)+6\ln m)\le 1/m^2$.

There are $m$ edges in total, so by union bound the probability of \emph{any} edge having congestion larger than $\mathcal{O}(\cong(P,d)+\log n)$ is strictly less than 1. 
\end{proof}

Having the tool of randomized rounding at our disposal, we now turn to the most involved lemma in our construction. If we want to route small demands we can already do so using Lemmata \ref{lem:krvtopath} (to get a path system) and \ref{lem:randomizer2} (to pick a small number of paths to store). However, routing a demand of size, say, $W$ from node $u$ to $v$, we would have to pick $2^{-l}W$ paths from a path system $P_{u,v}$ connecting $u$ and $v$, to ensure low congestion. Here, $\ell$ is the class of $P_{u,v}$.

Hence, if $\ell$ is small we would need to route a large number of paths. Instead, we find a cut consisting only of small (i.e., class $\ell$) edges separating each path in $P_{u,v}$.
\footnote{Note that this is not a cut of the \emph{graph}, which might still be connected, but of the \emph{path system}.}
These can be used to store routing information.

So we take all pairs $(u_i,v_i)$ in the same situation, that is, connected by a class $\ell$ path system, and take the union of all the cuts consisting of class $\ell$ edges. Then we route a single-commodity flow from the nodes $u_i$ to this cut. Of course, the packets from $u$ may have ended up at an edge belonging to some other $u_i$, so it may not be possible to route to $v$ directly. Instead we send the packets back through the single-commodity flow.

Again, the packets from $u$ may now reside in a different node $u'$. However, on the way they passed through a class $\ell$ edge, which we can use for storing the path from $u'$ to $v$. (To be precise, we use one of the adjacent nodes for storage.) But now we have a new problem---while both $P_{u,v}$ and $P_{u',v'}$ are class $\ell$ path systems, $P_{u',v}$ need not be. If it has class at least $\ell$, all is well and we can route the packet with a single path. Though if it has a class $\ell'<\ell$ we have to route using multiple paths again.

The fact that the class keeps decreasing allows us to solve this problem recursively. At each class we split the packet into smaller ones and find an edge to store the routing information for each of them. This stops when the packet is small enough to route directly, at the latest once it has reached size 1.

When we refer to storing the routing information for a node $u$ in some previous node $v$ on the path of a packet, we are using shorthand for a slightly elaborate transformation scheme, which we will refer to as \emph{anticipative routing}. When the packet arrives at node $v$, the node checks the packet header and adds the stored routing information to it, before sending the packet on its way normally. Then, once the packet has reached the node $u$ the routing information is extracted from the packet header and used.

\begin{lemma}[General graph embedding]\label{lem:krmtoperm}
Let $G'=(V, A, d)$ denote a weighted, directed graph, where the total weight of incoming and outgoing arcs of a node $v$ is at most $\we(v)$, and $\deg_{G'}(v)\in\mathcal{O}(\deg(v)\log^2n)$. Then there is a CTS that routes $\RouteArg[2.5pt]{d(u,v)}{\one_v}{\one_v}$ for each $(u,v)\in A$ with congestion $\mathcal{O}(\Co\log^2n\log^2W)$.

The routing table of node $v$ has size $\mathcal{O}(\deg(v)\Co\log^2n\log W\log^3(nW))$, while packet headers have length $\mathcal{O}(\log^3(nW))$. Commodity $(u,v)\in A$ is encoded as $\ell\in\{1,...,\deg_{G'}(u)\}$.
\end{lemma}
\begin{proof}
As mentioned above, the problematic demands are those which need multiple paths to route with low congestion. We will refer to those demands as \emph{large}. More precisely, we call an arc $(u,v)\in A$ \emph{$l$-large} if $d(u,v)>2^\ell$ and $\ell$ is the class of $P_{u,v}$ (defined below).

The proof will proceed in three parts.
\begin{enumerate}[(a)]
\item First, we construct path systems $P=\{P_{u,v}: u,v\in V\}$, s.t.\ all paths in $P_{u,v}$ have the same class and can be routed by storing a $\mathcal{O}(\log^3(nW))$ path id. For any demands $d'$ where both the total incoming and outgoing demand of a node $v$ are at most $\we(v)$, we have $\cong(P,d')\in\mathcal{O}(\Co\log^2n\log W)$. 
\item Then we show that we can partially route arcs $a\in A$ which are $l$-large, replacing them with $2^{-l}d(u,v)$ arcs of weight $2^l$. This does not change either the total outgoing or incoming demand of any node.
\item Finally, we construct the CTS and derive the resulting bounds.
\end{enumerate}

\smallskip\emph{Part (a).} We use Lemma \ref{lem:krvtopath} to construct a \emph{deterministic} concurrent transformation scheme $\mathit{TS}_1$ routing the PMCF. Hence we can have $P'_{u,v}$ denote the path system containing the paths from $u$ to $v$, one for each path id. Then we employ Valiant's trick and define $P^*_{u,v}$ as $\bigcup_{w\in V}P_{u,w}'\circ P_{w,v}'$, where $P\circ P'$ is a concatenation of path systems $P,P'$ given by $P\circ P':=\{p\circ p':p\in P,p'\in P'\}$. That means that we can split a path in $P^*_{u,v}$ into its first and second part.

Now consider some demands $d'$, where the incoming or outgoing demand of any node $v$ is at most $\we(v)$. As $\bigcup_{w\in V}P_{u,w}'$ are all outgoing paths of $u$, sampling one u.a.r.\ is equivalent to sending a packet with $\mathit{TS}_1$ from $u$. So the first parts create the same congestion as $\mathit{TS}_1$, given that $\sum_vd(u,v)\le\we(v)$. To be precise, the congestion increases by a factor of $1+\mathcal{O}(n^{-1})$, the approximation guaranteed by Lemma \ref{lem:krvtopath}. This is only a constant factor, so we are going to disregard it.

The intermediate node $w$ follows distribution $\we$. The second parts, i.e., the paths from $P_{w,v}'$ then have weight $\sum_u\we(w)d(u,v)\le\we(w)\we(v)$. Routing a packet from $w$ using $\mathit{TS}_1$ chooses a path from $P_{w,v}'$ with weight $\bwe(v)$, so we also bound the congestion based on $\mathit{TS}_1$.

In total we get $\cong(P^*,d')\le 2\cong(\mathit{TS}_1,\we)\in\mathcal{O}(\Co\log^2n)$. Finally, we want to modify $P^*_{u,v}$ so that it only contains paths of one class. We simply pick a class $\ell$ with the maximum number of paths in $P^*_{u,v}$, and set $P_{u,v}:=\{p\in P^*_{u,v}:p\text{ has class $l$}\}$. As there are $\Nclass$ classes, we have $|P^*_{u,v}|\le|P_{u,v}|\Nclass$ and the congestion increases by $\mathcal{O}(\log W)$.

Each path in $P_{u,v}$ is the concatenation of two paths from $\mathit{TS}_1$, so we can store two path ids of $\mathit{TS}_1$ to route it.

\smallskip\emph{Part (b).} We choose the highest class $\ell$ where the set $A'$ of $\ell$-large arcs is non-empty. Additionally, we introduce $\operatorname{str}:A\rightarrow V$, which is the node that will be used to store the routing information for an arc. Initially, $\operatorname{str}(u,v)=u$.

For any arc $a\in A'$ we define $d'(a)$ as the largest multiple of $2^\ell$ s.t.\ $d'(a)\le d(a)$, and then set $d:=d-d'$. So when routing $a$, a coin is flipped. With weight $d(a)$ we route using $a$ (how precisely is yet to be determined), with weight $d'(a)$ we do the following procedure.

For all $(u,v)\in A'$ the path system $P_{u,v}$ has class $\ell$. Using Lemma \ref{lem:randomizer2} with demands $d'$, we find a set of $d'(u,v)2^{-l}$ class $\ell$ paths from $u$ to $v$ for $(u,v)\in A$, with congestion $\mathcal{O}(\Co\log^2n\log W)$. We let $M$ denote the set of prefixes of these paths, up to (and including) their first class $\ell$ edge. By treating $M$ them as a flow $f$, we can construct a transformation scheme $\mathit{TS}$, using Lemma \ref{lem:sctopath}, which has the same congestion.

Going back to the arc $a$ we want to route, we send a packet using $\mathit{TS}$, with path id chosen uniformly at random. The necessary information for this is stored in $\operatorname{str}(a)$. There are $N:=d'(a)2^{-l}$ paths in $M$ for $a$ (and thus path ids of $\mathit{TS}$). If we put that number of tokens into the source of $a$ and apply $\mathit{TS}$ (one path id per token), they end up at nodes $z_1,...,z_N$. A node $z_i$ may receive tokens from other demands $a'\in A'$, but at most $\mathcal{O}(\deg(z_i)\Co\log^2n\log W)$ in total, as each path ending in $z_i$ in $M$ induces a load of $2^\ell$ on a class $\ell$ edge adjacent to $z_i$, i.e., a congestion of 1.

We also construct a transformation scheme based on the reverse flow $-f$, to send the tokens back. This does \emph{not} use a random path id, instead a node $z_i$ stores a mapping from incoming to outgoing path ids (any mapping is fine). We remark that the tokens of $a$ may not end up where they started, as routing through a single-commodity flow mixes packets arbitrarily. 

To summarize, an arc $a=:(u,v)$ has sent out $N$ tokens, each of which corresponds to $2^\ell$ flow from $d'(a)$. Each token traversed an intermediate node $z_i$ to end up at a node $u'$. Node $z_i$ was passed by a low number of tokens in total. So now we add a new arc $a':=(u',v)$ to $A$, with demand $d(a'):=2^\ell$. Crucially, the routing information for $a'$ is stored in $z_i$, i.e., $\operatorname{str}(a'):=z_i$.

As a technical detail we note that we allow for parallel arcs in $G'$. It is important that we do not merge multiple small demands into a larger one, as we have already ensured sufficient storage space for each, which would be lost.

The tokens are routed through $f$ and then $-f$, so the number of tokens starting and ending at $u$ is the same. This implies that both the total outgoing and incoming demand of any node remain unchanged.

As mentioned above, this procedure uses anticipative routing. For demand $a$ we send $N$ packets from $u$, each of which follows a deterministic path. So the intermediate node $z_i$ assigns the packet the specific path id sending it to $u'$ as well as the (yet to be determined) information on how to proceed from there. At $u'$ the node does not have to look up the packet header in its routing table, but merely execute the information contained within.

\smallskip\emph{Part (c).} First, we apply (b) at most $\Nclass$ times to eliminate all large arcs. Note that while (b) introduces new arcs, these have demand $2^\ell$, where $\ell$ is maximum class s.t.\ $\ell$-large arcs exist. So the new demands can only be $\ell'$-large for an $\ell'<\ell$.

Now we route the remaining arcs. Those are not large, so we can use Lemma \ref{lem:randomizer2} to pick a single path from $P_{u,v}$ for each $(u,v)\in A$. Based on our construction in (a), each path in $P_{u,v}$ can be routed using a $\mathcal{O}(\log^3(nW))$ path id. This will be stored in $\operatorname{str}(u,v)$.

For the initial arcs, we store their path ids within their respective source nodes together with their (encoded) commodity.

Finally, we analyze the congestion and space requirements.

Each use of (b) creates congestion of $\mathcal{O}(\Co\log^2n\log W)$, due to embedding two flows. Routing the non-large arcs at the end creates the same congestion (though only once). So in total we have a congestion of $\mathcal{O}(\Co\log^2n\log^2W)$.

In total, each node $v$ is used at most $\deg_{G'}(v)\in\mathcal{O}(\deg(v)\log^2n)$ times for storage due to our initial demands, and then at most $\cdot\mathcal{O}(\Co\log^2n\log W)$ times for each adjacent class $\ell$ edge when executing (b) for class $\ell$. Storing routing information for a large arc needs $\mathcal{O}(\log (nW))$ additional space to store the number of tokens and the range of path ids for them. This is dominated by the $\mathcal{O}(\log^3(nW))$ sized path id we need for both large and non-large arcs. (Recall that a large demand is first split into a fractional part and a multiple of $2^\ell$.)

To embed the flows in (b) using Lemma \ref{lem:sctopath}, we need a total of $\mathcal{O}(\deg(v)\log (nW)\log W)$ space per node $v$, and transformation scheme in (a) from Lemma \ref{lem:krvtopath} uses $\mathcal{O}(\deg(v)\log^3(nW)))$ space. Summing everything up, we get $\mathcal{O}(\deg(v)\Co\log^2n\log^2W\log^3(nW))$.

Regarding packet headers, we need packet headers of Lemmata \ref{lem:sctopath} and \ref{lem:krvtopath}, as well as some additional space for our anticipative routing (at most $\mathcal{O}(\log^3(nW))$). In total we get $\mathcal{O}(\log^3(nW))$.
\end{proof}

We want to remark on a slight technicality in the previous proof. Usually, scaling the routed distributions by some constant factor will scale the congestion by the same and nothing of importance has changed. However, the proof argues that there is a bound on the space used for each node, based on the congestion. Scaling the routed distribution to decrease congestion does actually affect this bound, so we could try scaling the congestion even lower. Though, as it turns out it is not possible to get a congestion below $\mathcal{O}(\Co\log^2n\log W)$ as that is the minimum when fixing a single path provided by part (a). Using a fractional path would indeed have lower congestion, but not take up less space.

\paragraph*{Hypercube embedding}
Now we move on to the hypercube embedding. Consider some cluster $S$ with children $S_1,...,S_r$. The general idea is that we assign each node $v$ some $w_S$ hypercube ids, by giving each child cluster $S_i$ an interval of $w_S(S_i)=\operatorname{out}_{S_i}(S_i)$ hypercube ids, distributed according to $\operatorname{out}_{S_i}$. (Recall that $w_S=\sum_i \operatorname{out}_{S_i}$.) Of course, this does not quite make a hypercube, so we have to skew the distributions by at most some constant factor so that everything ends up in a power of two.

\begin{figure}[t]
\begin{center}
\def\svgwidth{10.7cm}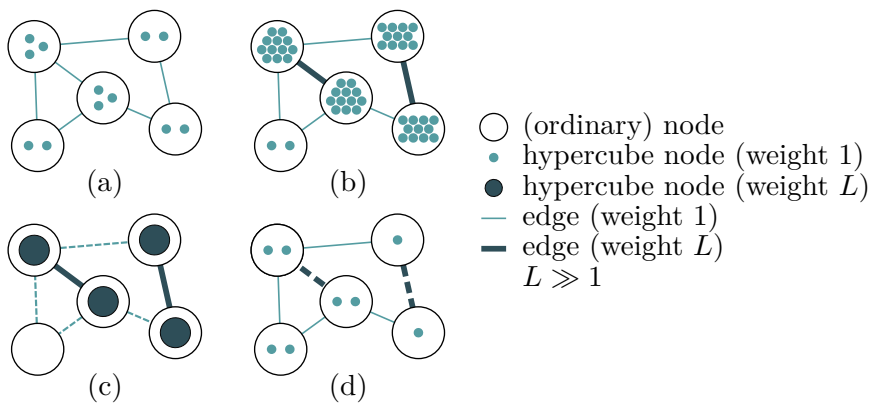
\end{center}
\caption{Embedding a hypercube in a cluster $S$. If all edges have weight 1, assigning hypercube nodes according to $w_S$ ensures that no node receives more than its degree (a). However, if edges with large weights exist, this no longer works (b). Instead, we assign hypercube nodes according to edges of one class, here either edges with weight $L$ (c) or weight 1 (d).}
\label{fig:hypercube}
\end{figure}

Then there is a second problem, illustrated in Figure \ref{fig:hypercube}. The main reason for embedding a hypercube is to reduce the number of paths a node has to store to reach any target. Within the hypercube, each node has logarithmic degree but we can still route to any node due to the special structure of the hypercube. We then leverage that to route to the interval assigned to child cluster $S_i$, effectively routing to $\operatorname{out}_{S_i}$.

However, we need to ensure that a node $v$ is assigned $\tildeO(\operatorname{deg}(v))$ ids. If there are much more than that, we cannot store routing information in $v$ for its adjacent hypercube edges. The original result in \cite{raecke2018compact} had $w_S(v)\le\operatorname{deg}(v)$ due to unit weights, but we do not. Instead we will assign a node $v$ roughly $w^{(\ell)}_S(v)/2^\ell$ hypercube ids, meaning one for each adjacent class $\ell$ edge contributing to $w_S$.

There are, of course, at most $\operatorname{deg}(v)$ such edges, so we do not run into storage problems. But if the distributions $w^{(\ell)}_S$ and $w_S$ are too dissimilar then we cannot route between them with low congestion. (The PMCF only ensures that distributions close to $w_S$ can be routed well.) Hence we need to choose the class $\ell$ carefully, so that $w^{(\ell)}_S(S_i)$ contains enough edges and we do not have to put too much flow on any single node.

This creates another complication, as different child clusters may necessitate different choices of $\ell$. While child clusters may have different classes, there are only $\Nclass$ many classes. So we will implement a hypercube for each of them, and later have a flow for each class which sends the data into the initial distribution for the specific hypercube.

\begin{lemma}[Hypercube embedding]\label{lem:routinghyp}
Let $S$ be an arbitrary cluster with children $S_1,...,S_r$. There exists a compact CTS that routes $\RouteArg[3pt]{w_S(S_i)}{\operatorname{maj}^{(\ell)}_S}{\operatorname{out}^{(\ell)}_{S_i}}$ for each $S_i$ of class $\ell$ with approximation $2$ and congestion $\mathcal{O}(\Co\log^3n\log^3W)$. 

The routing table of node $v$ has size $\mathcal{O}(\deg(v)\Co\log^2n\log W\log^3(nW))$, while packet headers have length $\mathcal{O}(\log^3(nW))$. There exists a numbering of $S_1,...,S_r$ s.t.\ commodity $S_i$ is encoded as integer $i$.
\end{lemma}
\begin{proof}
In the same manner as Räcke and Schmid~\cite{raecke2018compact}, we embed a hypercube. However, we use a hypercube for each class $\ell$ of edges and each hypercube id we assign has weight $2^\ell$, i.e., a node $v\in S$ gets roughly $\operatorname{maj}_S^{(\ell)}(v)/2^\ell$ hypercube ids.

The construction proceeds in a similar manner as the one of Räcke and Schmidt up until the embedding of the hypercube edges, where we use Lemma \ref{lem:krmtoperm} instead of simple randomized rounding. We start by arguing that we can renumber the child clusters s.t.\ given the index we can determine both the class and the approximate weight of a child.
\smallskip

\emph{Storing child class and weight.}
For a child cluster $S_i$ with class $\ell$, let $\norm{S_i}$ denote the smallest power of two with $\operatorname{out}_S^{(\ell)}(S_i)/2^\ell\le\norm{S_i}\le 2\,\operatorname{out}_S^{(\ell)}(S_i)/2^\ell$. This is the number of hypercube nodes that we assign to $S_i$.

We store the number of children of each class, which takes $\mathcal{O}(\log r\log W)$ bits, in each node in $S$. Additionally, for each class $\ell$ we store the number of child cluster of that class which have a specific value of $\norm{S_i}$. There are at most $1+\log_2m$ different values for $\norm{S_i}$, so we need $\mathcal{O}(\log r\log n)$ bits. In total, this uses $\mathcal{O}(\log r\log n\log W)$ bits in each node in $S$.

For our renumbering, we sort child clusters $S_i$ by class and, within a class, by their value of $\norm{S_i}$. Given an index $i$ based on this sorting, we can determine both class of $S_i$ and $\norm{S_i}$.
\smallskip

\emph{Constructing the class $l$ hypercube.} 
Fix some class $\ell$. We will now describe the construction of the class $\ell$ hypercube, then analyze at the end the congestion for all classes at once.

The hypercube has dimension $d$, with $d\in\mathbb{N}$ minimal s.t.\
$2^d\ge\sum_{L(S_i)=\ell}\norm{S_i}$, where $L(S_i)$ is the class of $S_i$. Each
class $\ell$ child $S_i$ gets a range of $\norm{S_i}$ ids, distributed such that a
node $v\in S_i$ gets between $\operatorname{out}_S^{(\ell)}(v)/2^l$ and
$2\operatorname{out}_S^{(\ell)}(v)/2^\ell$ hypercube ids. These ids are stored in
$v$. As the order of children is fixed and stored within each node, we can
recompute the range of any child cluster during routing.

We have assigned $\sum_{i=1}^{r}\norm{S_i}$ hypercube ids in total, which may be less than $2^d$. Hence we distribute the other hypercube ids evenly across the nodes of class $l$ child clusters $S_i$, s.t.\ a node $v\in S_i$ receives at most $2\operatorname{out}_S^{(\ell)}(v)/2^\ell$ additional hypercube ids, and thus between $\operatorname{out}_S^{(\ell)}(v)/2^\ell$ and $4\operatorname{out}_S^{(\ell)}(v)/2^\ell$ in total. These other hypercube ids will only be used during routing as intermediate nodes.
\smallskip

\emph{Congestion within the hypercube.} 
Now consider some packet at a node $u\in S$ that we want to route to $S_i$. First we pick a hypercube node $x$ u.a.r.\ among those assigned to $u$ (they are stored in $u$). Then we pick a hypercube node $y$ u.a.r.\ from the range assigned to $S_i$ (which we can recompute). Then we route from $x$ to $z$, a random intermediate node in the hypercube, then from $z$ to $y$. 

We remark that, in a hypercube, the PMCF with weights $\we:=1$ can be solved with congestion $\mathcal{O}(1)$ and that this bound is achieved by routing in the usual manner, i.e., fixing an order for the bits and sending the packet along the edge according to the first bit different between source and target. As we are using Valiant's trick, the congestion is determined by the maximum incoming or outgoing amount of flow for a single node.

For the congestion, we consider routing $\RouteArg[4pt]{\operatorname{out}^{(\ell)}_S(S_i)}{\operatorname{maj}^{(\ell)}_S}{\operatorname{out}^{(\ell)}_{S_i}}$ for $S_i$ with class $\ell$, i.e.\ $\operatorname{out}^{(\ell)}_{S_i}(S_i)$ units of flow instead of $w_S(S_i)=\operatorname{out}_{S_i}(S_i)$. As $S_i$ has class $\ell$, we have $\operatorname{out}^{(\ell)}_{S_i}(S_i)\ge\operatorname{out}_{S_i}(S_i)/\Nclass$ and the congestion increases by a factor of at most $\Nclass$.

Summing up all $\operatorname{out}^{(\ell)}_S(S_i)$ we get $\operatorname{maj}^{(\ell)}_S(S)$, so a node $v\in S_i$ sends out $\operatorname{out}_S^{(\ell)}(v)$ packets, and each hypercube node $x$ of $v$ sends at most $2^\ell$ of them. For commodity $i$ there are $\norm{S_i}\ge\operatorname{out}_S^{(\ell)}(S_i)/2^\ell$ hypercube nodes, so each receives at most $2^\ell$ packets.

Both outgoing and incoming flow of a hypercube node are at most $2^\ell$, so the load on a hypercube edge is also at most $\mathcal{O}(2^\ell)$.

While we send to a hypercube node from the range of $S_i$ u.a.r., a node $v\in S_i$ is assigned between $\operatorname{out}_S^{(\ell)}(v)/2^\ell$ and $2\operatorname{out}_S^{(\ell)}(v)/2^\ell$ of them. Hence the target distribution is only within an approximation of 2.
\smallskip

\emph{Embedding into the original graph.} 
Finally, we embed the hypercube using Lemma \ref{lem:krmtoperm}. A node $v\in S_i$ for $S_i$ of class $\ell$ has at most $4\operatorname{maj}_S^{(\ell)}(v)/2^\ell\le 4\deg(v)$ hypercube ids. So there are at most $8m$ nodes in the hypercube in total, and the degree of each node is $\mathcal{O}(\log n)$. Let $d_\ell(u,v)$ denote the number of edges connecting $u$ and $v$ in the class $\ell$ hypercube, for $u,v\in S$, and $d:=\sum_\ell 2^\ell d_\ell$. Setting $A:=\{(u,v):d(u,v)>0\}$ we embed the graph $G':=(S,A,d)$.

As the load on an edge of the class $\ell$ hypercube is at most $2^\ell$, in total $d(u,v)$ packets are sent from $u$ to $v$. A node $v$ in a class $\ell$ child has outgoing and incoming demand at most $\mathcal{O}(k2^\ell\log n)\subseteq\mathcal{O}(w_S(v)\log n)$, where $k$ is the number of class $l$ edges incident to $v$. The congestion of Lemma \ref{lem:krmtoperm} increases by $\mathcal{O}(\log W)$ due to decreasing the total number of packets earlier in our analysis, and $\mathcal{O}(\log n)$ due to the outgoing and incoming demand of a node.

While we use an additional $\mathcal{O}(\log r\log n\log W)$ space per node $v$ to store the sizes of clusters, and $\mathcal{O}(\deg(v)\log n)$ to store the hypercube ids of nodes assigned to $v$, this is dominated by the cost of Lemma \ref{lem:krmtoperm}, which also determines the sizes of packet headers.
\end{proof}

\paragraph*{Unmixing CTS}
Given the hypercube embedding from the last lemma, we can now construct the unmixing CTS. At the beginning we need to ensure that we move to the distribution for the correct class, then we move through the (class specific) hypercube, and finally we go to the target distribution.

\begin{lemma}[Unmixing CTS]\label{lem:unmixing}
There exists a CTS that routes $\RouteArg[3pt]{w_S(S_i)}{w_S}{w_{S_i}}$ for each $i=1,...,r$ with congestion $\mathcal{O}(\Co\log^3n\log^3W)$. The routing table of node $v$ has size $\mathcal{O}(\deg(v)\Co\log^2n\log W\log^3(nW))$, while packet headers have length $\mathcal{O}(\log^3(nW))$. There exists a numbering of $S_1,...,S_r$ s.t.\ commodity $S_i$ is encoded as integer $i$.
\end{lemma}
\begin{proof}
The numbering of child clusters and our path ids are the same as for Lemma \ref{lem:routinghyp}. Therefore we can determine the class $\ell$ of $S_i$ based on its index, as shown in the proof of that lemma.

For a child $S_i$ with class $\ell$ we want to route $\Route{w_S}{\operatorname{maj}_S^{(\ell)}}\Route{}{\operatorname{out}_{S_i}^{(\ell)}}\Route{}{\operatorname{out}_{S_i}}$.

\begin{enumerate}[(1)]
\item For each class $l$ let $M\subseteq S$ denote the union of class $l$ child clusters. We route $\RouteArg[3pt]{w_S(M)}{w_S}{\operatorname{maj}_S^{(\ell)}}$ using Lemma \ref{lem:routingsimilar} with congestion $\Co\cdot w_S(M)/\operatorname{maj}_S^{(\ell)}(M)$. This is at most $\Co\Nclass$, as $\operatorname{maj}_{S}^{(\ell)}(S_i)=\operatorname{out}_{S_i}^{(\ell)}(S_i)\ge w_S(S_i)/\Nclass$ for each child $S_i$ with class $l$.
\item We use Lemma \ref{lem:routinghyp} once, to route $\RouteArg[3pt]{w_S(S_i)}{\operatorname{maj}^{(\ell)}_S}{\operatorname{out}^{(\ell)}_{S_i}}$, with congestion $\mathcal{O}(\Co\log^3n\log^3W)$.
\item For each $S_i$ we route $\RouteArg[3pt]{w_S(S_i)}{\operatorname{out}_{S_i}^{(\ell)}}{\operatorname{out}_{S_i}}$ \emph{within $S_i$} using Lemma \ref{lem:routingsimilar}. Here we have congestion $\Co\cdot w_S(S_i)/\operatorname{out}_{S_i}^{(\ell)}(S_i)\le \Co\Nclass$.
\end{enumerate}

Note that (1) has to be implemented on the whole cluster for each class, so its total congestion is $\mathcal{O}(\Co\log^2W)$ (but still lower than step (2)). For the bounds on space per node and length of packet headers, the costs of step (2) dominate.
\end{proof}

\subsection{Combining the Results}\label{sec:combining}

Lemma \ref{lem:routinghyp} can be used directly as a drop-in replacement in the original result in \cite{raecke2018compact}. However, we have organized things slightly differently and thus feel it necessary to repeat the analysis.

The key idea is routing between two nodes $u$ and $v$ using the decomposition tree, spreading out a packet according to distribution $w_S$ in each cluster. This ensures that routing within a cluster can be done with low congestion. Moving through the tree, the congestion is determined by the bottlenecks $\operatorname{out}_S$. However, the optimal algorithm has to send the packets through these bottlenecks as well, so we remain competitive.

\begin{theorem}\label{thm:finalresult}
There exists a compact oblivious routing scheme with competitive ratio $\mathcal{O}(\log^6n\log^3W)$, using a routing table of length $\mathcal{O}(\deg(v)\log^5n\log W\log^3(nW))$ for a node $v\in V$, packet headers of length $\mathcal{O}(\log^3(nW))$, and node labels of length at most $\mathcal{O}(\height(T)\log\deg(T))$.
\end{theorem}
\begin{proof}
The analysis is mostly analogous to \cite[Lemma 2]{raecke2018compact}, apart from the slight change that Lemmata \ref{lem:mixing} and \ref{lem:unmixing} route directly between $w_S$ and $w_{S_i}$ instead of splitting into an upper and lower sub-path.

To route from node $u$ to $v$, we determine the clusters $S_u,S_v$ containing just these nodes, i.e., $S_u=\{u\}$ and $S_v=\{v\}$. Let $p$ denote the path in the decomposition tree from $\{u\}$ to $\{v\}$, which has length $k\in\mathcal{O}(\log n)$. We start in distribution $\bar{w}_{S_u}(V)=\one_u$, and want to end at $\bar{w}_{S_v}(V)=\one_v$. This is done by going through the sequence of distributions $w_{p_1},w_{p_2},...,w_{p_k}$, routing from $w_{p_i}$ to $w_{p_{i+1}}$ using Lemma \ref{lem:mixing} if $p_{i+1}$ is the parent of $p_i$, and Lemma \ref{lem:unmixing} otherwise. 

We accumulate a slight multiplicative error of $1+\mathcal{O}(n^{-1})$ at each step, which is bounded by a constant factor in total, as we have at most $2\,\height(T)\in\mathcal{O}(\log n)$ steps. The final distribution is $\one_v$ and remains unchanged by any error, so this merely increases congestion by a constant.

It is necessary to determine the path through the decomposition tree, hence the label of a node $v$ consists of the path in the decomposition tree, encoded as a sequence of child cluster indices (given by Lemma \ref{lem:routinghyp}. These are enough to determine the full path, by looking at the node labels of the start and end node.

Now we analyze the competitive ratio. Let $d:V\times V\rightarrow\mathbb{R}$ denote demands.

Fix any edge $e\in E$. Load on $e$ is generated only when routing between distributions $w_{S_i}$ and $w_{S}$ for some cluster $S$ with child cluster $S_i$, where $S$ contains both endpoints of $e$. This uses that the routing between the two distributions happens inside of $S$, and does not generate load on any edge not fully contained. Sending a packet from $u$ to $v$ involves routing between distributions $w_{S_i}$ and $w_{S}$ only if one of $u,v$ is not in $S_i$ and the other one is, so the total demand for these is $\lambda(i):=\sum_{u\in S_i}\sum_{v\notin S_i}(d(u,v)+d(v,u))$.

However, the demand $\lambda(i)$ must enter or leave $S_i$ (and thus pass over an edge in $\operatorname{out}_{S_i}$) regardless of our specific routing scheme. So there are $\lambda(i)\le C_{\operatorname{opt}}\operatorname{out}_{S_i}(S_i)$ such packets at most, where $C_{\operatorname{opt}}$ is the optimal congestion for demands $d$. Using $w_S(S_i)=\operatorname{out}_{S_i}(S_i)$ we get $w_S(S_i)/\lambda(i)\le C_{\operatorname{opt}}$.

Applying Lemmata \ref{lem:mixing} and \ref{lem:unmixing} with $\Co\in\mathcal{O}(\log^2n)$ then results in a congestion of at most $\mathcal{O}(C_{\operatorname{opt}}\log^5n\log^3W)$, and for each node $v\in S$ it uses $\mathcal{O}(\deg(v)\log^4n\log W\log^3(nW))$ space, as well as packet headers of length $\mathcal{O}(\log^3(nW))$.

Both edges and nodes can be contained in at most $T_\mathrm{h}\in\mathcal{O}(\log n)$ clusters, giving the final bounds on congestion and space per node. For a packet we need to store the path through the decomposition tree, so $\mathcal{O}(\log n)$ path ids of length $\mathcal{O}(\log\deg(T))$ and the length of a packet header does not increase.

As mentioned above, we store the cluster indices in the label of a node $v$, for each cluster in which $v$ is contained, resulting in node labels of length $\mathcal{O}(T_\mathrm{h}\log\deg(T))$.
\end{proof}

\begin{corollary}\label{col:finalresultsmall}
Assume $W\in\mathcal{O}(\operatorname{poly}(n))$. Then there exists a compact oblivious routing scheme with competitive ratio $\mathcal{O}(\log^9n)$, using a routing table of length $\mathcal{O}(\deg(v)\log^9n)$ for a node $v\in V$, packet headers of length $\mathcal{O}(\log^3n)$ and node labels of length $\mathcal{O}(\log^2n)$.
\end{corollary}

%% file: figures/hypercube.pdf_tex
\begingroup%
  \makeatletter%
  \providecommand\color[2][]{%
    \errmessage{(Inkscape) Color is used for the text in Inkscape, but the package 'color.sty' is not loaded}%
    \renewcommand\color[2][]{}%
  }%
  \providecommand\transparent[1]{%
    \errmessage{(Inkscape) Transparency is used (non-zero) for the text in Inkscape, but the package 'transparent.sty' is not loaded}%
    \renewcommand\transparent[1]{}%
  }%
  \providecommand\rotatebox[2]{#2}%
  \newcommand*\fsize{\dimexpr\f@size pt\relax}%
  \newcommand*\lineheight[1]{\fontsize{\fsize}{#1\fsize}\selectfont}%
  \ifx\svgwidth\undefined%
    \setlength{\unitlength}{416.69291339bp}%
    \ifx\svgscale\undefined%
      \relax%
    \else%
      \setlength{\unitlength}{\unitlength * \real{\svgscale}}%
    \fi%
  \else%
    \setlength{\unitlength}{\svgwidth}%
  \fi%
  \global\let\svgwidth\undefined%
  \global\let\svgscale\undefined%
  \makeatother%
  \begin{picture}(1,0.54421769)%
    \lineheight{1}%
    \setlength\tabcolsep{0pt}%
    \put(0,0){\includegraphics[width=\unitlength,page=1]{hypercube.pdf}}%
    \put(0.14574431,0.28587943){\color[rgb]{0,0,0}\makebox(0,0)[t]{\lineheight{1.25}\smash{\begin{tabular}[t]{c}(a)\end{tabular}}}}%
    \put(0.44508699,0.28587943){\color[rgb]{0,0,0}\makebox(0,0)[t]{\lineheight{1.25}\smash{\begin{tabular}[t]{c}(b)\end{tabular}}}}%
    \put(0.14574431,0.03350789){\color[rgb]{0,0,0}\makebox(0,0)[t]{\lineheight{1.25}\smash{\begin{tabular}[t]{c}(c)\end{tabular}}}}%
    \put(0.44508699,0.03350789){\color[rgb]{0,0,0}\makebox(0,0)[t]{\lineheight{1.25}\smash{\begin{tabular}[t]{c}(d)\end{tabular}}}}%
    \put(0,0){\includegraphics[width=\unitlength,page=2]{hypercube.pdf}}%
    \put(0.65944061,0.35737312){\color[rgb]{0,0,0}\makebox(0,0)[lt]{\lineheight{1.25}\smash{\begin{tabular}[t]{l}(ordinary) node\end{tabular}}}}%
    \put(0.65944061,0.31972596){\color[rgb]{0,0,0}\makebox(0,0)[lt]{\lineheight{1.25}\smash{\begin{tabular}[t]{l}hypercube node (weight 1)\end{tabular}}}}%
    \put(0.65944061,0.2820788){\color[rgb]{0,0,0}\makebox(0,0)[lt]{\lineheight{1.25}\smash{\begin{tabular}[t]{l}hypercube node (weight $L$)\end{tabular}}}}%
    \put(0.65944061,0.24443164){\color[rgb]{0,0,0}\makebox(0,0)[lt]{\lineheight{1.25}\smash{\begin{tabular}[t]{l}edge (weight 1)\end{tabular}}}}%
    \put(0.65944061,0.20678448){\color[rgb]{0,0,0}\makebox(0,0)[lt]{\lineheight{1.25}\smash{\begin{tabular}[t]{l}edge (weight $L$)\end{tabular}}}}%
    \put(0.65944061,0.16913732){\color[rgb]{0,0,0}\makebox(0,0)[lt]{\lineheight{1.25}\smash{\begin{tabular}[t]{l}$L\gg 1$\end{tabular}}}}%
  \end{picture}%
\endgroup%